\documentclass[11pt]{article}

\usepackage[letterpaper,margin=1in]{geometry}
\usepackage[T1]{fontenc}
\usepackage{lmodern}
\usepackage{microtype}
\usepackage{amsmath,amssymb,bm,amsthm,mathtools}
\usepackage{booktabs}
\usepackage{array}
\usepackage{graphicx}
\usepackage{float}
\usepackage{xcolor}
\usepackage{tikz}
\usetikzlibrary{arrows.meta,positioning,calc}
\usepackage[authoryear]{natbib}
\usepackage[hidelinks]{hyperref}
\usepackage{dsfont}

\usepackage{enumitem} 

\IfFileExists{algorithm.sty}{\usepackage{algorithm}}{}
\IfFileExists{algpseudocode.sty}{\usepackage[noend]{algpseudocode}}{}
\usepackage{caption}

\usepackage{comment}

\def \oo {\infty}

\def \E {\mathbb{E}}

\def \P {\mathbb{P}}

\def \ra {\rightarrow}
\def \to {\rightarrow}

\def \pf {\noindent \emph{Proof: }}
\def \QED {\hfill $\square$}

\def \RV {\mbox{\rm{RV}}}

\newtheorem{rmk}{\bf Remark}

\newtheorem{lemma}{\bf Lemma}
\newtheorem{prop}{\bf Proposition}

\newcommand{\CM}{\mathcal{CM}}
\newcommand{\BF}{\mathcal{BF}}
\newcommand{\CBF}{\mathcal{CBF}}
\newcommand{\St}{\mathcal{S}}
\newcommand{\arcosh}{\operatorname{arcosh}}

\newcommand{\rgm}[1]{\texttt{#1}}
\newcommand{\lU}{\lambda_U}
\newcommand{\lL}{\lambda_L}
\newcommand{\kU}{\kappa_U}
\newcommand{\kL}{\kappa_L}

\allowdisplaybreaks

\title{A new tractable Archimedean copula for full-range tail dependence}
\author{Lei Hua\thanks{School of Mathematical and Statistical Sciences, Northern Illinois University, United States. Email: Lhua@niu.edu}}
\date{\today}

\begin{document}
\maketitle

\begin{abstract}
Parametric copulas capable of capturing full-range tail dependence in both lower and upper tails are useful in many applications, including as pair-copula building blocks in vine copula models. However, it is challenging to create such copulas that have closed-form density functions, which limits their practical applicability in complex modeling scenarios. To the best of our knowledge, all existing full-range tail dependence copulas do not have closed-form density functions. In this paper, we propose a new Archimedean copula family, called the \emph{full-range Archimedean copula of type I} (FRA1), that captures full-range tail dependence in both tails and admits a closed-form density function. The copula has two parameters, one controlling the strength of dependence in the lower tail and the other controlling that in the upper tail. We provide results on dependence properties of the copula, and discuss how to effectively estimate the parameters based on maximum likelihood. We build a vine copula model only using the FRA1 copula, and demonstrate that the model can achieve comparable performance of an ordinary vine copula model while avoiding the need of pair-copula model selection. The new copula is also applied to data from the Medical Expenditure Panel Survey, demonstrating that an FRA1 vine model can effectively capture diverse temporal and cross-sectional dependence patterns in medical expenditures. Finally, the new copula is implemented in \texttt{CopulaOne}\footnote{https://github.com/larryleihua/CopulaOne}, an R package for full-range tail dependence copulas.
\end{abstract}    


\section{Introduction}
Tail dependence patterns that can be captured by a parametric copula often matter the most in applications, especially in fields such as insurance, finance, and quantitative risk management. Many bivariate parametric copulas have been proposed in the literature, and they are widely used mainly because of the availability of their closed-form density functions, which allow for constructing relatively more complex models while performing parameter estimation using the maximum likelihood method. However, most existing parametric copulas can only model very specific dependence patterns in the tails: either asymptotic dependence or asymptotic independence; either being symmetric or asymmetric between the lower and upper tails. To address this issue, some parametric copulas have been created to model full-range tail dependence, such as the GGEE copula proposed in \cite{Hua2017}, and the PPPP copula proposed in \cite{Su2017b}. However, these copulas do not have closed-form density function, although the PPPP copula is very close to achieve that, for which a numerical root-finding is needed to evaluate the density function. Although estimation can be performed using these copulas, it is currently still inconvenient to use them to build more sophisticated models, such as vine copula models, or those include not only copulas but also complex marginal models, unless suitable likelihood-free estimation methods are further developed, for which we refer to a recent effort in \cite{Hua2026c} for the GGEE and the PPPP copulas. That being said, it is still desirable to have a bivariate copula that can model full-range tail dependence in both lower and upper tails and has a closed-form density function, which is the main goal of this paper.

We construct the new copula based on the Archimedean copula family, which has been widely studied and popular choices for many applications. Dependence patterns in the tails of an Archimedean copula can be analyzed through the behavior of its generator function $\psi(s)$ near zero and infinity. The key idea of the new Archimedean copula is to design a generator function that allows for flexible tail behavior in both limits, thereby enabling full-range tail dependence in both the lower and upper tails. In the literature, several methods have been studied to create Archimedean generators from suitable functional forms. Among others, \cite{Joe1996, Joe1997} use mixtures of positive powers of max-infinitely divisible distributions to create new Archimedean generators; \cite{Bacigal2015} studies aggregation of Archimedean generators to create new copulas; \cite{Kelner2021} employs a compound of an existing generator to construct new Archimedean copulas; \cite{Schilling2012} provides a comprehensive treatment of Bernstein functions that are useful to preserve functional properties of Archimedean generators; \cite{DiBernardino2017} proposes an upper-patching method to modify tail behavior of Archimedean generators; \cite{Gorecki2021} applies outer-power transformation to introduce more parameters for Archimedean copulas.

To construct Archimedean copulas that can capture very flexible dependence patterns in both lower and upper tails, respectively, we need at least two parameters. Among many different methods of introducing additional parameters in Archimedean generators, we consider creating a generator in the form of $\psi_{\eta, \theta} = G_\eta \circ W_\theta$. Then the Archimedean copula can be written as
\begin{align}
C_{\eta, \theta}(u,v)=\psi_{\eta, \theta}(\psi_{\eta, \theta}^{-1}(u)+\psi_{\eta, \theta}^{-1}(v)), \quad u, v \in [0,1]. \notag
\end{align}
We let $G_\eta(s)$ dominate the right tail of $\psi_{\eta, \theta}(s)$ and thus the lower tail of $C_{\eta, \theta}$; similarly, let $W_\theta(s)$ dominate the left tail of $\psi_{\eta, \theta}(s)$ and thus the upper tail of $C_{\eta, \theta}$. We refer to \cite{Charpentier2009} for a comprehensive study on how tail behavior of Archimedean generators affect the dependence patterns of the corresponding copulas. With carefully chosen functions of $G$ and $W$, we have achieved a new Archimedean copula that we call the \emph{full-range Archimedean copula of type I}, referred to as the \emph{FRA1 copula} in what follows; to keep the name short, ``full-range'' instead of ``full-range tail dependence'' used here is only for its lower and upper tail dependence behavior. We use type I here, because we believe that there may be other types of such Archimedean copulas. However, we will reserve such naming convention only for those highly selected Archimedean copulas that might be discovered in the future. The FRA1 copula has the following merits:
\begin{enumerate}
\item (Flexibility) It can capture full-range tail dependence in both lower and upper tails.
\item (Closed-form) It has elementary closed-form functions for the following: copula density function, copula cumulative distribution function (CDF), the $h$-function: $h(u,v):=\P(U \leq u \mid V=v)$, and the Blomqvist's beta.
\item (Parsimony \& Interpretability) The parameters $\eta$ and $\theta$ control the strength of dependence in the lower and upper tails, respectively. 
\item (Monotonicity) The strength of dependence in the lower tail is monotonically increasing in $\eta \in [-1,1)$, and the strength of dependence in the upper tail is monotonically increasing in $\theta \in [-1,1)$, with $(\eta, \theta)=(-1, -1)$ leading exactly to the independence copula.
\end{enumerate}

Because $\psi_{\eta, \theta}$ is completely monotone, it can be used to define an Archimedean copula in any finite dimension. The construction of the generator depends on many nice properties of completely monotone functions, and (complete) Bernstein functions play an important role in ensuring these properties. We will discuss in detail the construction and properties of $\psi_{\eta, \theta}$, and provide explicit closed-form formulas for the associated copula density, CDF, $h$-function, etc. Dependence properties of the resulting copula will also be thoroughly analyzed, and they include results on lower and upper tail orders and tail dependence coefficients, Kendall's tau, Spearman's rho, Blomqvist's beta, etc. 

Vine copulas provide a flexible framework for constructing high-dimensional distributions through bivariate copulas \citep{Bedford2002,Aas2009,Czado2019}. In addition to serving as an ordinary pair-copula for building up vine copulas, bivariate FRA1 copulas can be used solely to build vine copulas. Because of its flexibility in modeling all kinds of tail dependence structures, as well as its closed-form density functions, $h$-functions, and so on, the FRA1 copula is particularly suitable for building high-dimensional vine copula models. Then such a vine copula built with only the FRA1 copula may lead to better interpretation while avoiding the model selection that is typically required when using a variety of different bivariate copulas as the building blocks. We will build a C-vine structure using the FRA1 copula as the sole bivariate building block, and compare its performance with vine copulas constructed from a mix of different bivariate copulas. Simulation study tells us that the FRA1 vine can provide comparable performance in terms of overall model fit and tail dependence capture.

For real applications, the FRA1 copula is particularly useful when there are many different dependence structures that need to be captured simultaneously, or when the dependence structures are unknown and need to be estimated from data. We use the FRA1 vine copula to model the temporal and cross-sectional dependence structures in a medical expenditure survey data, and find some interesting dependence patterns that would be difficult to capture and interpret using traditional vine copulas with a mix of different bivariate copulas.

The remainder of this paper is organized as follows. In Sections \ref{sec:FR} and \ref{sec:archimedean}, we briefly review the concepts of full-range tail dependence and Archimedean copulas. In Section \ref{sec:model}, we introduce the construction of the FRA1 copula and discuss its closed-form status. Section \ref{sec:dependence-properties} provides results for the dependence properties of the FRA1 copula. Section \ref{sec:MLE} discusses the nonregular issue of the likelihood function and how to conduct maximum likelihood estimation accordingly.
Section \ref{sec:vine} presents the application of the FRA1 copula in building vine copulas and its comparison to ordinary vine copulas. Section \ref{sec:empirical} contains the empirical study with medical expenditure survey data. Finally, Section \ref{sec:conclude} concludes the paper, and all proofs are reported in the Appendix.

\subsection{Full-range tail dependence}
\label{sec:FR}
Let $C$ be a bivariate copula and let $\widehat C(u,v) =u + v-1+C(1-u,1-v)$ be its survival copula. Then the lower and upper tail dependence parameters/coefficients are defined as
$\lambda_L =\lim_{u\downarrow0}C(u,u)/u, \lambda_U =\lim_{u\downarrow0} \widehat C(u,u)/u,$ whenever the limits exist. Asymptotic dependence is referred to the case when the tail dependence coefficient is positive, and asymptotic independent is when the coefficient is zero.

Under asymptotic independence, suppose that the diagonal tail probabilities
are regularly varying \citep{Bingham1987}: $C(u,u) =u^{\kappa_L}\ell_L(u), \widehat C(u,u) =u^{\kappa_U}\ell_U(u), u\downarrow0$, where $\ell_L$ and $\ell_U$ are slowly varying at zero.  Then $\kappa_L$ and $\kappa_U$ are referred to as the lower and upper tail orders; see \cite{Hua2011b} for the concept of tail order.

A copula having full-range tail dependence basically means that the values of $\lambda$ can be exhausted in the range of $[0,1]$, and the values of $\kappa$ can be exhausted in the range of $[1,2]$ within a single bivariate parametric copula; refer to Definition 2.1 of \cite{Su2017b} for a formal definition. The GGEE and the PPPP copulas are two examples that have full-range tail dependence in both lower and upper tails.

\subsection{Archimedean copulas and complete monotonicity}
\label{sec:archimedean}
A function $f:(0,\oo)\to \mathbb{R}$ is \emph{completely monotone}, denoted as $f\in\CM$, if $f \in C^\infty((0,\oo))$, the set of all infinitely differentiable functions defined on the open interval $(0, \oo)$, and $(-1)^n f^{(n)}(t) \geq 0, \forall t >0$ for all $n =0,1,2,\ldots$. Let $\psi:[0,\oo)\to (0,1]$ be continuous, strictly decreasing, satisfy $\psi(0)=1$ and $\psi(t) \downarrow 0$ as $t\to\oo$, and be completely monotone. Writing $\phi=\psi^{-1}$, the associated Archimedean copula is
\begin{align}
C(u,v) &=\psi(\phi(u)+\phi(v)). \label{eq:archimedean-copula-general}
\end{align}
Complete monotonicity is sufficient for Eq.~\eqref{eq:archimedean-copula-general} to define an Archimedean copula in every dimension \citep{McNeil2009}. Such $\psi$ is also the Laplace transform (LT) of a nonnegative random variable. For an Archimedean copula, the behavior of $\psi(t)$ as $t\downarrow 0$ determines the upper tail, whereas the decay of $\psi(t)$ as $t\to\oo$ determines the lower tail \citep{Charpentier2009}.

The complete monotonicity of an Archimedean generator plays an important role in the construction, and in Appendix \ref{appendix:generator-properties}, we provide some useful results on complete monotonicity and related results, which will be used to prove the complete monotonicity of the proposed Archimedean generator; refer to \cite{Schilling2012} for details.

\section{The FRA1 copula}
\label{sec:model}
In this section, we introduce the construction of the FRA1 copula and discuss the closed-form status of various components involved in the construction.

\subsection{The Archimedean generator}
To construct a bivariate Archimedean copula that has full-range tail dependence in both lower and upper tails, we first need to have a flexible Archimedean generator that has at least two parameters, one controlling the lower tail and the other controlling the upper tail. There are many closure properties related to completely monotone functions, and they provide a convenient way to construct a flexible Archimedean generator. In this paper, we focus on the following construction, where two parameters $\eta$ and $\theta$ are introduced to control the lower and upper tails, respectively:
\begin{align}
\psi_{\eta,\theta}(t) =G_\eta(W_\theta(t)).\label{eq:generator}
\end{align}
Then the FRA1 copula will be written as
\begin{align}
C_{\eta,\theta}(u,v) &=\psi_{\eta,\theta}(\phi_{\eta,\theta}(u)+\phi_{\eta,\theta}(v)). \label{eq:archimedean-copula}
\end{align}

\begin{prop}
\label{prop:generator-properties}
Suppose that $G_\eta$ and $W_\theta$ in Eq.~\eqref{eq:generator} are defined as:
\begin{align}
G_\eta(s) &= 
\begin{cases}
\displaystyle \exp\{-F_{-\eta}(s)\}, \quad -1\leq\eta\leq0,\\[2mm]
\displaystyle \left(1+\frac{F_\eta(s)}{b_\eta}\right)^{-b_\eta},  \quad b_\eta =(1-\eta)\eta^{-2}, \quad 0<\eta<1,
\end{cases} \label{eq:G-definition} \\
W_\theta(t)
&= \begin{cases}
J_{-\theta}(t), \quad -1\leq\theta\leq0,\\[1mm]
J_0(K_{1-\theta}(t)), \quad 0<\theta<1,
\end{cases} \label{eq:W-definition}
\end{align}
where
$F_p(s) = p^{-2}[\cosh(pA(s))-1],  0 \leq p \leq 1,$
with $F_0(s) = [A(s)]^2/2$ being understood as the limit as $p \downarrow 0$; $A(s) =\arcosh(1+s) = \log(s+1+\sqrt{s(s+2)}), s \geq 0$;
\begin{align}
J_p(t) =
\begin{cases}
(S_p(1/t))^{-1},&t>0,  \\
0,                              &t=0,
\end{cases} \notag
\end{align}
with $S_p(x) = T_p(x)-d_p$, $T_p(x) = F_p^{-1}(F_p(x)+1)$, $d_p :=F_p^{-1}(1)$, and $K_a(t) =(1+t^a)^{1/a}-1, 0<a\leq 1$.
Then $\psi_{\eta,\theta}$ is a completely monotone function for $-1\leq\eta<1$ and $-1\leq\theta<1$, with $\psi_{\eta,\theta}(0) = 1$ and $\psi_{\eta,\theta}(\oo) = 0$, and thus, $C_{\eta,\theta}$ defined in Eq. (\ref{eq:archimedean-copula}) is a valid bivariate Archimedean copula.
\end{prop}

The proof of Proposition~\ref{prop:generator-properties} is provided in Appendix~\ref{appendix:generator-properties}. Note that, all the functions involved in the construction of the FRA1 copula are elementary functions whose inverse functions can also be expressed in closed form. 

\subsection{Closed-form status of the FRA1 copula}
The FRA1 copula has closed-form CDF, density function, $h$-function, and Blomqvist's beta, which makes it analytically tractable for various applications. Closed-form status for relevant quantities is summarized in Table~\ref{tab:closed-form-status}, and the exact closed-form expressions are provided in Appendix \ref{appendix:closed-form-expressions}.

\begin{table}[H]
\centering
\caption{Closed-form status of the FRA1 copula.}
\label{tab:closed-form-status}
\begin{tabular}{p{0.22\textwidth}p{0.3\textwidth}p{0.33\textwidth}}
\toprule
Quantity & Required analytic ingredients & Status\\
\midrule
CDF & $\psi$ and $\psi^{-1}$ & Closed-form: Eqs.~\eqref{eq:generator}\eqref{eq:psi-inverse} \\
Density & $\psi^{-1}$, $\psi'$, and $\psi''$ & Closed-form: Eqs.~\eqref{eq:psi-inverse}\eqref{eq:psi-derivatives} \\
$h$-function & $\psi^{-1}$ and $\psi'$ & Closed-form: Eqs.~\eqref{eq:psi-inverse}\eqref{eq:psi-derivatives}\\
Inverse $h$-function & additionally $(\psi')^{-1}$ & Root finder for $(\psi')^{-1}$: Eq.~\eqref{eq:inverse-h-formal}\\
Kendall's $\tau$ & integral of $t[\psi'(t)]^2$ & Scalar integral: Eq.~\eqref{eq:Kendall-tau}\\
Spearman's $\rho_S$ & integral involving $\psi$ and $\psi'$ & Bivariate integral: Eq.~\eqref{eq:Spearman-rho}\\
Blomqvist's $\beta_B$ & $\psi$ and $\psi^{-1}$ at $1/2$ & Closed-form:  Eq.~\eqref{eq:Blomqvist-beta}\\
\bottomrule
\end{tabular}
\end{table}

For the inverse $h$-function, in general, $(\psi_{\eta,\theta}')^{-1}$ does not have a closed-form expression and may need to be computed numerically using a root-finding algorithm. However, $\psi_{\eta,\theta}'$ is monotone, which ensures the uniqueness of the root, and a numerical solution can be reliably obtained. The inverse $h$-function is mainly used in simulation and sampling from the copula, and it is not used when evaluating the likelihood of the copula or vine copula models built from it. Many copulas implemented in the popular R package \texttt{rvinecopulib}, such as BB1 in \cite{Joe1997}, also do not have closed-form inverse $h$-functions but still allow for effective use in practice.

For assessing overall dependence, Blomqvist's beta provides a convenient closed-form measure, and Kendall's tau and Spearman's rho need to be evaluated via numerical integration. However, the patterns of strength of dependence explained by these measures are similar across different combinations of parameters (see Fig.~\ref{fig:dependence}), so that Blomqvist's beta can serve as a reliable proxy for overall dependence in practice. Moreover, usually a numerical evaluation of Kendall's tau and Spearman's rho is fairly straightforward and computationally feasible.

Therefore, the FRA1 copula is a very attractive option for practical applications, considering that the functions relevant for estimation and inference all have closed-form expressions, and moreover, the FRA1 copula has very flexible dependence structures and interpretable parameters, which will be discussed in Section~\ref{sec:dependence-properties}.

\section{Dependence properties}
\label{sec:dependence-properties}
In this section, we study the dependence properties of the FRA1 copula, including its tail order and tail dependence coefficients, Kendall's tau,  Spearman's rho, and Blomqvist's beta.

\medskip

The following result establishes that the strength of dependence in the upper tail is monotonically increasing with respect to the parameter $\theta$ from $-1$ to $1$, while the strength of dependence in the lower tail is monotonically increasing with respect to the parameter $\eta$ from $-1$ to $1$.

\begin{prop}
\label{prop:tail-map}
For the FRA1 copula $C$ defined in Eq.~\eqref{eq:archimedean-copula}, the upper tail dependence coefficient and the upper tail order depend only on $\theta$ and satisfy
\begin{align}
\lambda_U(\theta)  = \max\{2-2^{1-\theta},0\}, \quad \kappa_U(\theta) =\max\{1-\theta,1\}, \quad -1\le\theta<1, \notag
\end{align}
and the lower tail dependence coefficient and the lower tail order depend only on $\eta$ and satisfy
\begin{align}
\lambda_L(\eta) = \begin{cases}
0, & -1\leq\eta\leq0,\\[1mm]
2^{-(1-\eta)/\eta}, & 0<\eta<1,
\end{cases} \qquad \kappa_L(\eta) = \max\{2^{-\eta},1\}, \quad -1\le\eta<1. \notag
\end{align}
\end{prop}

The proof of Proposition~\ref{prop:tail-map} is provided in Appendix~\ref{appendix:dependence}. Note that, the monotonicity of the tail dependence coefficients and tail orders with respect to the parameters $\theta$ and $\eta$ implies that increasing $\theta$ strengthens the tail dependence in the upper tail, while increasing $\eta$ strengthens the dependence in the lower tail.

\begin{rmk} \rm
The lower and upper tails are not symmetric even if they have the same strength of dependence in terms of tail order or tail dependence coefficient. This is not the limitation introduced by how the generator is constructed, but rather an inherent feature of Archimedean copulas. For Archimedean copulas, only bivariate Frank, the independence, and the countermonotonic copulas are reflection symmetric.
\end{rmk}

\begin{rmk} \rm
When $\eta = \theta = -1$, the copula reduces to the independence copula. Note that, in this case $F_1(s) = \cosh(\arcosh(1+s))-1=s$, $T_1(x) = x+1$, $d_1 = 1$, $S_1(x)=x$, and $J_1(t) = t$. Therefore, $\psi_{-1,-1}(t) = e^{-t}$, which is the generator of the independence copula.
\end{rmk}  

Figure~\ref{fig:contour} shows contour plots of the dependence structure of the FRA1 copula $C_{\eta,\theta}$, with different strengths of dependence in lower and upper tails, in terms of tail order ($\kappa$) and tail dependence coefficient ($\lambda$). Note that, unlike the GGEE copula \citep{Hua2017} and the PPPP \citep{Su2017b}, the FRA1 copula is not reflection symmetric (i.e., $C_{\eta,\theta}(u,v) \ne \widehat{C}_{\eta,\theta}(v,u)$), although it can capture full-range tail dependence in both lower and upper tails and the patterns between lower and upper tails are quite similar (see Figure~\ref{fig:contour}).

\begin{figure}[H]
\caption{Contour plot of the FRA1 copula with normal scores as margins.}
\label{fig:contour}
\includegraphics[width=\textwidth]{\detokenize{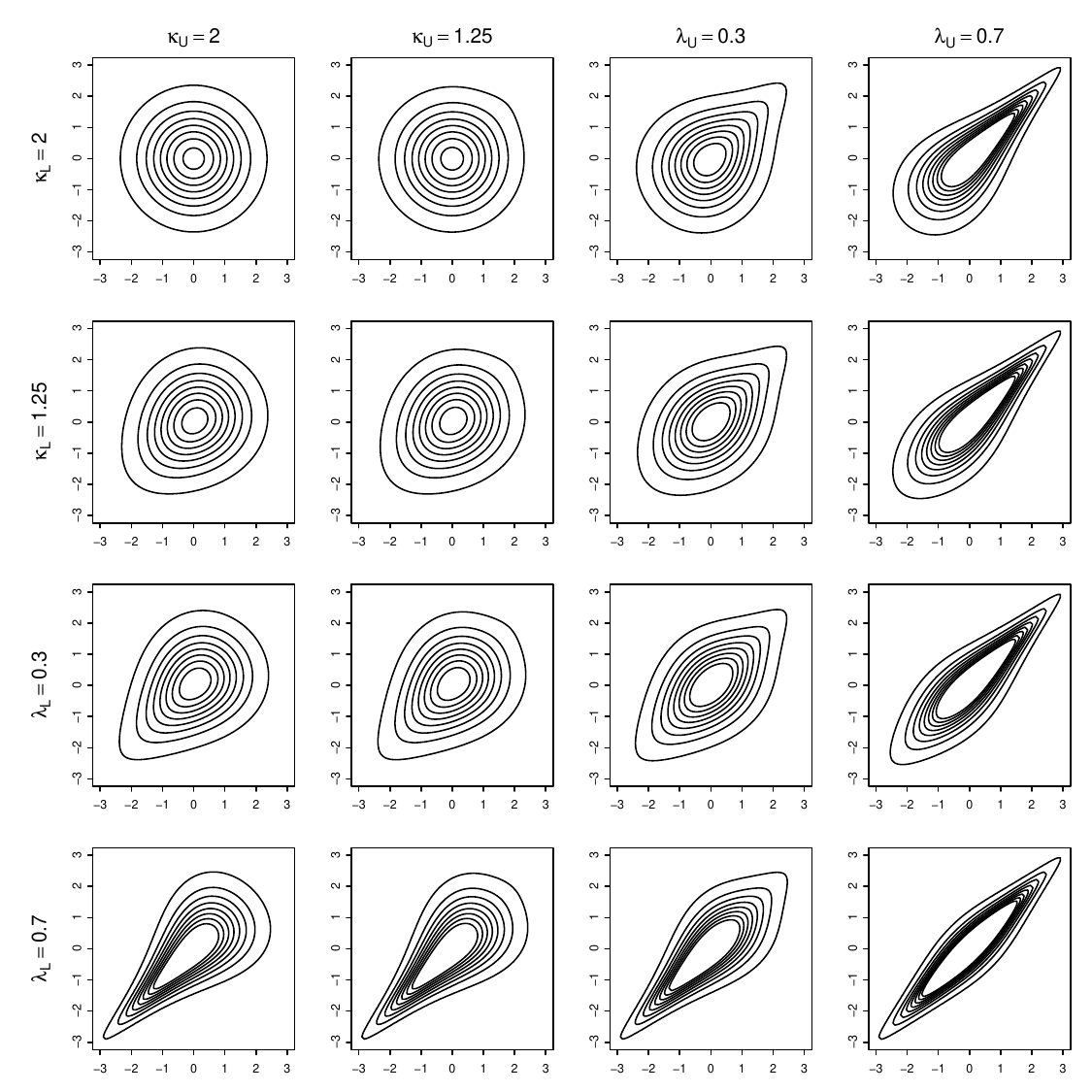}}
\end{figure}

\begin{prop}
\label{prop:Kendall-tau}
Let $C_{\eta,\theta}$ be the Archimedean copula defined by the generator $\psi_{\eta,\theta}$. Then the Kendall's tau of $C_{\eta,\theta}$ is given by
\begin{align}
\tau(C_{\eta,\theta}) =1-4\int_0^\infty t[\psi_{\eta,\theta}'(t)]^2 d t. \label{eq:Kendall-tau}
\end{align}
\end{prop}

\begin{prop}
\label{prop:Spearman-rho}
Let $C_{\eta,\theta}$ be the Archimedean copula defined by the generator $\psi_{\eta,\theta}$. Then the Spearman's rho of $C_{\eta,\theta}$ is given by
\begin{align}
\rho_S(C_{\eta,\theta}) = 12 \int_0^\infty \int_0^\infty \psi_{\eta,\theta} (x+y) \psi_{\eta,\theta}'(x)
\psi_{\eta,\theta}'(y) d x d y - 3.  \label{eq:Spearman-rho}
\end{align}
\end{prop}

\begin{prop}
\label{prop:Blomqvist-beta}
Let $C_{\eta,\theta}$ be the Archimedean copula defined by the generator $\psi_{\eta,\theta}$. Then the Blomqvist's beta of $C_{\eta,\theta}$ is given by
\begin{align}
\beta_B(C_{\eta,\theta}) =4\psi_{\eta,\theta} (2\phi_{\eta,\theta}\left( 1/2 \right) )-1. \label{eq:Blomqvist-beta}
\end{align}  
\end{prop}

The proofs of Propositions \ref{prop:Kendall-tau}, \ref{prop:Spearman-rho}, and \ref{prop:Blomqvist-beta} are provided in the Appendix \ref{appendix:additional-proofs}. Note that, only Blomqvist's beta has a closed-form expression, while Kendall's tau and Spearman's rho involve integrals. Figure \ref{fig:dependence} illustrates the dependence structure in terms of these common measures of dependence. Note that, while the exact numerical values may differ, the overall patterns of dependence are generally consistent across these measures. Also, the overall dependence measures may not be monotone in the parameters of the copula. For example, when $\eta$ takes a negative value, say, -0.1, Spearman's rho may initially increase with $\theta$ and then decrease, reflecting the non-monotone relationship. However, it does not mean that the copula is not identifiable, and it simply means that the copula is not identifiable by these overall dependence measures. The copula is identifiable through the two parameters, as they are designed particularly to capture dependence structures in the lower and upper tails, respectively, and strength of dependence in the tails is monotone in the parameters.

\begin{figure}[H]
\caption{Dependence structure illustration.}
\label{fig:dependence}
\includegraphics[width=\textwidth]{\detokenize{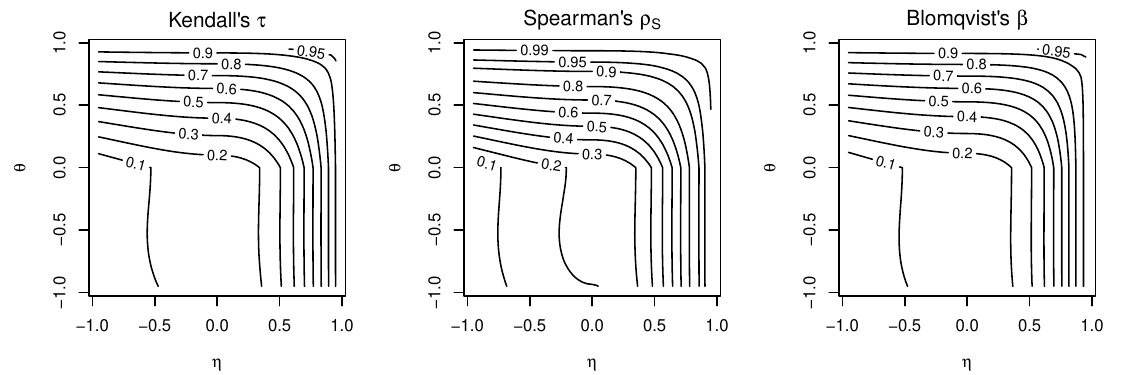}}
\end{figure}

\section{Maximum likelihood estimation}
\label{sec:MLE}
From the definition of $G_\eta(s)$ in \eqref{eq:G-definition} and $W_\theta(s)$ in \eqref{eq:W-definition}, we notice that the Archimedean generator is stitched together at $(0,0)$ from four pieces of $(\eta, \theta) \in [-1, 1)^2$, that is, $[-1,0]\times[-1,0]$, $[-1,0]\times(0,1)$, $(0,1)\times[-1,0]$, and $(0,1)\times(0,1)$. This has been designed purposely to allow the parameters change continuously inside the parameter space, and to describe continuous change of strength of  dependence captured by the copula. 

Figure \ref{fig:MLE} illustrates this structure, using the standardized $\log(\psi)$ curve: $\Delta \log(\psi) (\eta, \theta):= \log(\psi_{\eta, \theta}(5)) - \log(\psi_{0, 0}(5))$, which is a function of $\eta$ and $\theta$. At $\eta=0$ and $\theta=0$, these two functions are nonregular, in the sense that given $\eta=0$, $\Delta \log(\psi)$ is not differentiable at $\theta=0$ (Figure \ref{fig:MLE}.C), and given $\theta=0$, $\Delta \log(\psi)$ is flat at $\eta=0$ (Figure \ref{fig:MLE}.D). 

\begin{figure}[H]
\includegraphics[width=\textwidth]{\detokenize{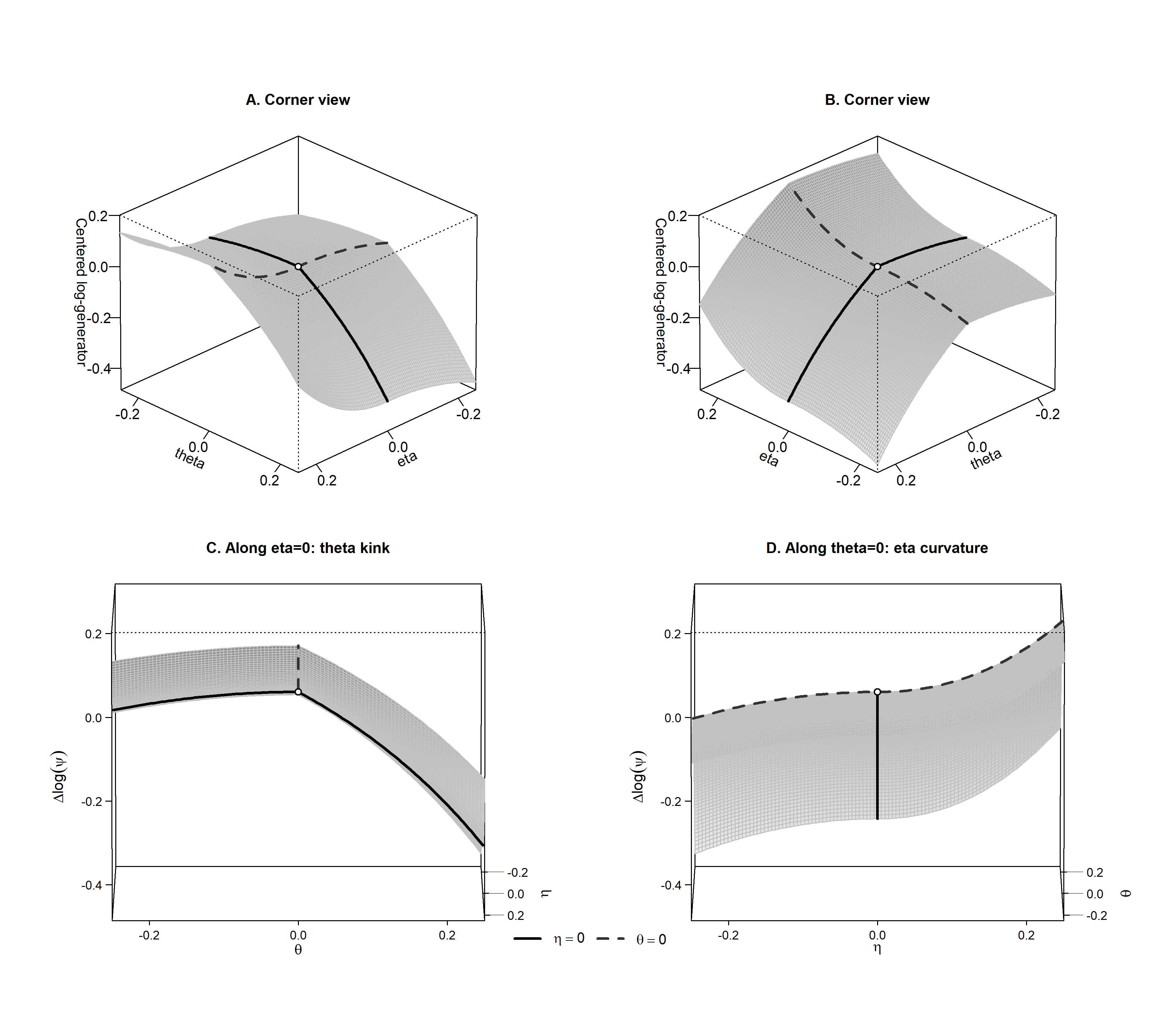}}
\caption{Illustration of $\Delta \log(\psi_{\eta, \theta}(t_0))$ as a function of $(\eta, \theta)$ from different viewpoints; $\eta, \theta \in [-0.25, 0.25], t_0=5$; scales are consistent across the plots. The solid lines represent the curve of $\Delta \log(\psi)$ when $\eta=0$, and the dashed lines represent the curve of $\Delta \log(\psi)$ when $\theta=0$.}
\label{fig:MLE}
\end{figure}

Because of the nonregularity, we cannot use standard optimization methods based on gradient or Hessian information on the whole parameter space. Instead, we can get the MLE on each of the four parameter regions first. Then the overall MLE can be obtained by comparing the likelihood values across these regions. All calculations involved are based on elementary closed-form formulas, and therefore, the speed for estimation is reasonably fast. When speed of estimation is critical, we can use a two-step procedure: first, obtain rough estimates on a coarse grid over the parameter space, and then refine the estimates using local optimization within the promising regions. There can be more sophisticated methods developed to further improve the efficiency and accuracy of the estimation process, but this is beyond the scope of the current paper and can be considered as future work.

\section{Full-range tail dependence for vine copulas}
\label{sec:vine}
The FRA1 copula has full-range tail dependence in both lower and upper tails, and more importantly, it has closed-form density functions, $h$-functions, etc., which makes it computationally convenient for many practical applications. 

Because of the full-range tail dependence and computational convenience, the FRA1 copula is particularly suitable for constructing vine copulas, as we can use only the FRA1 copula as the building block for all pair-copulas in the vine structure, avoiding the need to select different copula families for different edges. When negative dependence is needed, we simply use the $90$ degree rotation of the FRA1 copula. 

In this section, we conduct a simulation study to evaluate the performance of the FRA1 vine copula model compared to ordinary C-vine and R-vine models.

\subsection{Simulation setting}
We use the canonical vine (C-vine) structure for illustration in this paper, and compare it with ordinary C-vine and R-vine structures with ordinary parametric families implemented in the R package \texttt{rvinecopulib} \citep{Nagler2025a}. Table \ref{tab:models} summarizes the candidate models.

\begin{table}[H]
\centering
\caption{Candidate vine copula models.}
\label{tab:models}
\begin{tabular}{lp{0.8\textwidth}}
\toprule
Model & Specification \\
\midrule
FRA1 vine & The FRA1 copula with rotation 0 or 90 only; separate parameters for each edge \\
Ord. C-vine & Ordinary parametric families; same C-vine structure as the FRA1 vine model; edge families selected by BIC \\
Ord. R-vine &  Ordinary parametric families; R-vine structure and edge families selected from the data \\
Oracle  & True data-generating model; simulation-based noise-floor benchmark \\
\bottomrule
\end{tabular}
\end{table}
The ordinary C-vine is the primary comparison with the FRA1 vine because its structure is held equal to the FRA1 vine. The selected R-vine represents a practical ordinary-vine analysis in which structure selection is also allowed. The oracle is used as a benchmark for the best achievable performance.

We simulate four different scenarios for evaluating the performance of the FRA1 vine and competing vine copula models. These scenarios vary in terms of the underlying dependence structure, tail dependence patterns, and the presence of negative dependence. All scenarios have five dimensions, and for the C-vine, the root order is from 1 to 5, that is, the joint density is
\begin{align}
c(u_1, \dots, u_5) & = c_{12}(u_1, u_2) c_{13}(u_1, u_3) c_{14}(u_1, u_4) c_{15}(u_1, u_5) \notag \\
                  & \times c_{23|1}(u_2, u_3 \mid u_1) c_{24|1}(u_2, u_4 \mid u_1) c_{25|1}(u_2, u_5 \mid u_1) \notag \\
                  & \times c_{34|12}(u_3, u_4 \mid u_1, u_2) c_{35|12}(u_3, u_5 \mid u_1, u_2) \notag \\
                  & \times c_{45|123}(u_4, u_5 \mid u_1, u_2, u_3). \notag
\end{align}
Note that, with such a non-truncated C-vine, all pairs of variables are directly or indirectly connected through the vine structure, allowing for flexible modeling of dependencies between any two variables. The details of each scenario are listed in Table~\ref{tab:scenarios}.

\begin{table}[H]
\centering
\caption{Scenarios for simulation study.}
\label{tab:scenarios}
\begin{tabular}{lp{0.7\textwidth}}
\toprule
Scenario & Specification \\
\midrule
FR\_asymmetric & Use only unrotated FRA1 copula for C-vine, and $(\theta, \eta)$ alternates between $(0.55, -0.45)$ and $(-0.50, 0.45)$. \\
FR\_negative\_rotations & Use FRA1 copula for C-vine with every edge rotated by 90 degrees. The same alternating $(\theta, \eta)$ values as FR\_asymmetric are used. \\
ordinary\_mixed\_tail & Ordinary C-vine cycling through Clayton, Gumbel, Joe, Frank, Student-$t$(df=4), and BB1; all (conditional) Kendall taus are positive and decided by $0.46 - 0.075(t-1)$, where $t$ is the tree level. \\
ordinary\_mixed\_signed & The same ordinary mixed-family C-vine, with every third edge assigned a negative Kendall tau. Frank copula uses a signed parameter for negative dependence, while other negative edges use 90/270-degree rotations. \\
\bottomrule
\end{tabular}
\end{table}

\subsection{Evaluation criteria}
The following three goodness-of-fit criteria will be used to evaluate the fitted copula models: Kullback-Leibler (KL) gap from oracle, overall CDF root mean square error (RMSE), and all-corner tail log-ratio RMSE. These criteria capture different aspects of model fit, including divergence from the true copula, overall distributional accuracy, and capability in modeling tail behavior. All are loss measures, so smaller values are better.

\begin{itemize}
\item KL gap from oracle: It measures the divergence between the true copula and the fitted model. Let $c_0$ be the true copula density, $\widehat c_m$ a fitted candidate, and $\{\bm U_s^{(0)}:s=1,\ldots,N_T\}$ an independent test sample, where $N_T$ is the independent test sample size drawn from the true copula; $N_T = 5000$ is used throughout. Then the estimate of the KL gap for model $m$ is
\begin{equation}
\widehat\Delta_{\mathrm{KL},m} =\frac{1}{N_T}\sum_{s=1}^{N_T} \left[\log c_0\!\left(\bm U_s^{(0)}\right) -\log \widehat c_m\!\left(\bm U_s^{(0)}\right)\right]. \notag
\end{equation}
\item Overall CDF RMSE: It assesses the overall distributional accuracy of the fitted model. Let $\{\bm U_s^{(m)}:s=1,\ldots,N_M\}$ be an independent sample drawn from candidate $m$, where $N_M$ is the independent sample size for candidate $m$; $N_M = 5000$ is used throughout. At $K$ common points $\bm x_1,\ldots,\bm x_K$ sampled from the test data ($K=300$ used), define the empirical multivariate CDFs for the true copula and candidate $m$, respectively, as
\begin{align}
\widehat C_0(\bm x_k)
&=\frac{1}{N_T}\sum_{s=1}^{N_T}
 \mathds{1}\left(\bm U_s^{(0)}\leq\bm x_k\right), \quad \widehat C_m(\bm x_k) =\frac{1}{N_M}\sum_{s=1}^{N_M} \mathds{1}\!\left(\bm U_s^{(m)}\leq\bm x_k\right), \notag
\end{align}
where $\mathds{1}(\cdot)$ is the indicator function. Then, the estimated overall CDF RMSE for candidate $m$ is
\begin{align}
\widehat{\mathrm{RMSE}}_{C,m}
&=\left(\frac{1}{K}\sum_{k=1}^{K}
\left[\widehat C_m(\bm x_k)-\widehat C_0(\bm x_k)\right]^2 \right)^{1/2}. \notag
\end{align}
\item All-corner tail log-ratio RMSE: It evaluates the model's capability in capturing tail dependence. Let $\mathcal P=\{(i,j):1\leq i<j\leq5\}$, so $|\mathcal P|=10$. For threshold $q\in\mathcal Q$, where $\mathcal{Q} = \{0.05, 0.10, 0.15, 0.20\}$ is used, the four bivariate corner events are
\begin{align}
A_{ij}^{LL}(q)&=\{U_i\leq q,\ U_j\leq q\}, & A_{ij}^{UU}(q)&=\{U_i>1-q,\ U_j>1-q\}, \notag \\ 
A_{ij}^{LU}(q)&=\{U_i\leq q,\ U_j>1-q\}, & A_{ij}^{UL}(q)&=\{U_i>1-q,\ U_j\leq q\}. \notag
\end{align}
For source $a\in\{0,m\}$ and corner $h \in \mathcal{H} = \{\mathrm{LL},\mathrm{UU},\mathrm{LU},\mathrm{UL}\}$ defined by the threshold $q$, let $K_{a,ij}^{h}(q)$ be the observed event count in a sample of size $N_a$. Define the tail rate for corner $h$ as
\begin{align}
r_{a,ij}^{h}(q) =\frac{K_{a,ij}^{h}(q)+1/2}{(N_a+1)q}, \notag
\end{align}
where $1/2$ is used to prevent undefined log ratios in rare cells. Then the all-corner tail log-ratio RMSE for candidate $m$ is
\begin{align}
\widehat{\mathrm{RMSE}}_{\mathrm{tail},m} =\left(\frac{1}{4|\mathcal P||\mathcal Q|} \sum_{q\in\mathcal Q}\sum_{(i,j)\in\mathcal P} \sum_{h\in\mathcal H} \left[\log\left(\frac{r_{m,ij}^{h}(q)} {r_{0,ij}^{h}(q)}\right)\right]^2 \right)^{1/2}. \notag
\end{align}
\end{itemize}

\subsection{Performance}
The performance can be clearly compared using plots. Figures~\ref{fig:FR_asymmetric-n780}--\ref{fig:ordinary_mixed_tail-n780} display the above three goodness-of-fit criteria for the different candidate vine copulas under each scenario, where points are individual Monte Carlo replicates, and box plots summarize their distributions. For each box plot, there are 200 Monte Carlo replicates. The sample size for each scenario is $n=780$. We have also conducted the comparison with sample size $n=390$, and the results are qualitatively similar to those presented here, so are omitted for brevity.

In Figures \ref{fig:FR_asymmetric-n780} and \ref{fig:FR_negative_rotations-n780}, we can observe that the FRA1 vine model generally performs well across the three evaluation criteria. This is actually expected, as these two scenarios are based on the FRA1 vine copula model. This means, the FRA1 vine model can capture the underlying dependence structure accurately if it is correctly specified. The R-vine model is relatively worse because these scenarios are generated based on C-vine structures, while the R-vine model is more flexible but may not be as efficient in capturing the specific structure of the data.

In contrast, Figures \ref{fig:ordinary_mixed_signed-n780} and \ref{fig:ordinary_mixed_tail-n780} show scenarios where the FRA1 vine model may not have a clear advantage. These scenarios are generated from more complex dependence structures that are not specifically tailored to the FRA1 vine model. Again, R-vine models perform the worst as these two scenarios are based on C-vine structures, which the R-vine model may not capture efficiently.

In all the four scenarios, the FRA1 vine model shows comparable performance to the ordinary vine models with model selection and vine structure selection, not only in terms of the overall goodness-of-fit but also in the tail behavior, which is crucial for capturing various dependence patterns in the tail. 

These results are quite exciting as they demonstrate that vine models constructed by only the FRA1 copula can serve as an alternative to more complex vine structures, avoiding pair-copula selection and improving interpretability as the same copula is used consistently throughout the vine.

\begin{figure}[H]
\centering
\includegraphics[width=\textwidth]{\detokenize{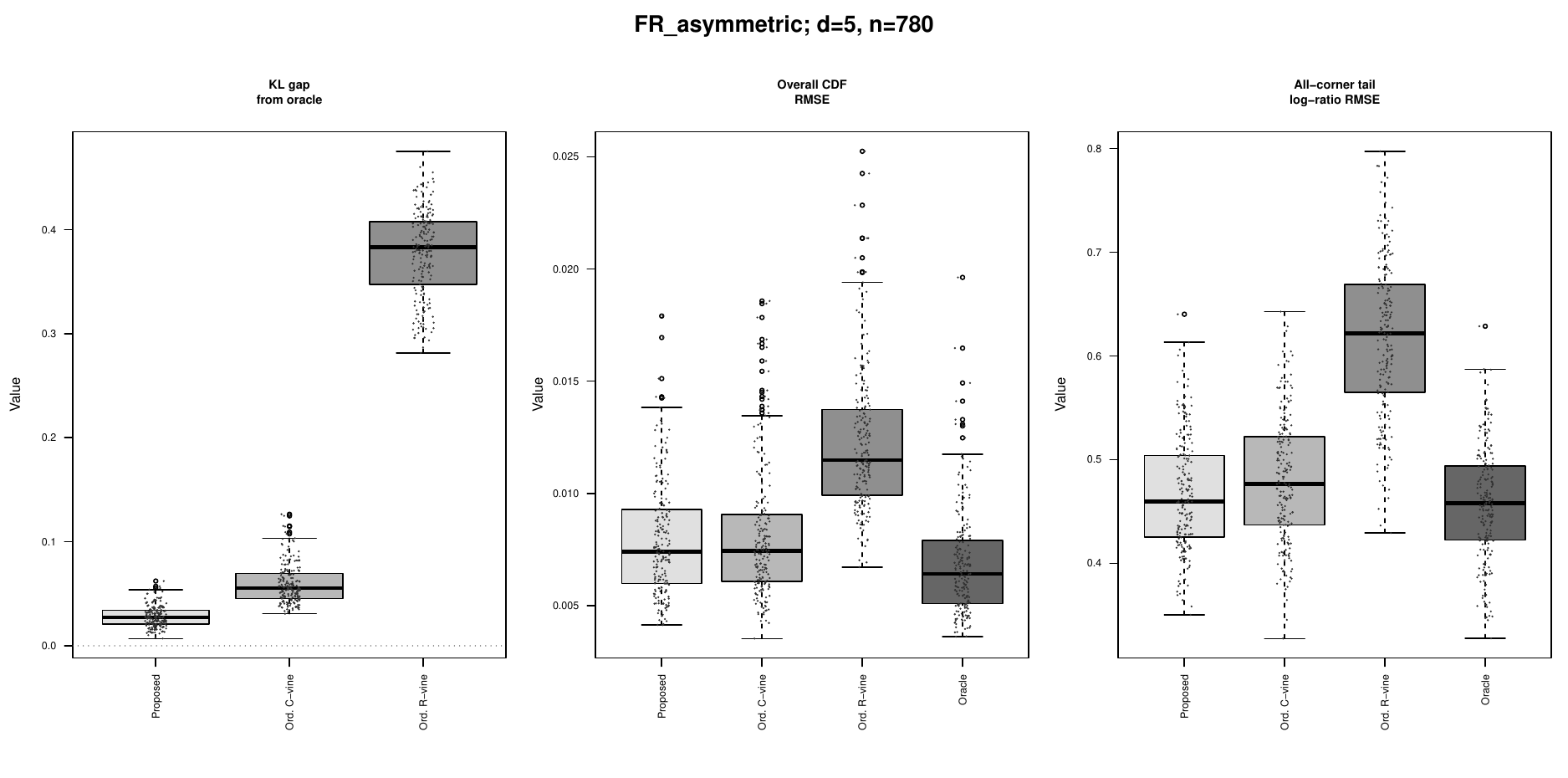}}
\caption{Three goodness-of-fit criteria for scenario FR\_asymmetric.}
\label{fig:FR_asymmetric-n780}
\end{figure}

\begin{figure}[H]
\centering
\includegraphics[width=\textwidth]{\detokenize{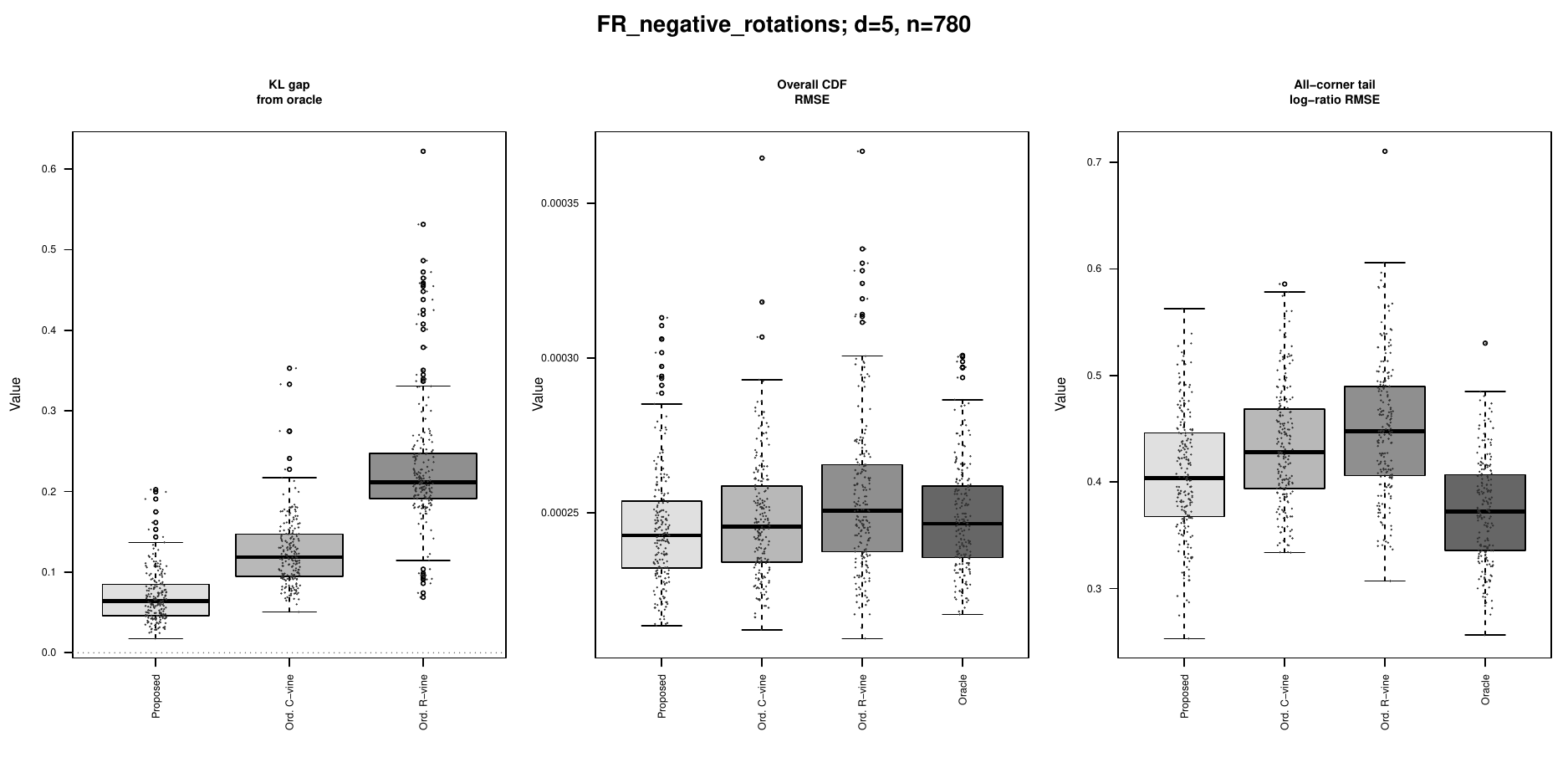}}
\caption{Three goodness-of-fit criteria for scenario FR\_negative\_rotations.}
\label{fig:FR_negative_rotations-n780}
\end{figure}

\begin{figure}[H]
\centering
\includegraphics[width=\textwidth]{\detokenize{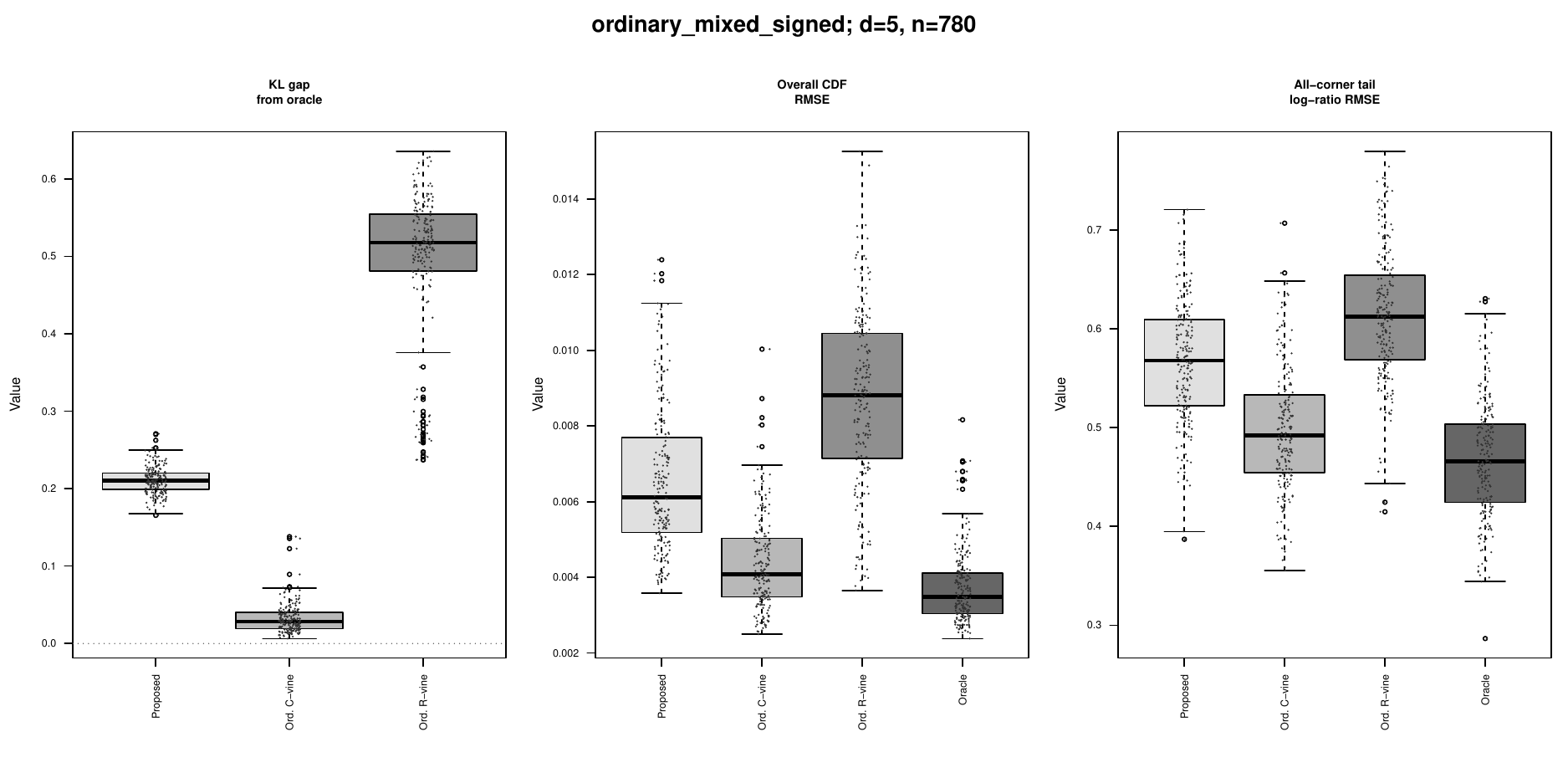}}
\caption{Three goodness-of-fit criteria for scenario ordinary\_mixed\_signed.}
\label{fig:ordinary_mixed_signed-n780}
\end{figure}

\begin{figure}[H]
\centering
\includegraphics[width=\textwidth]{\detokenize{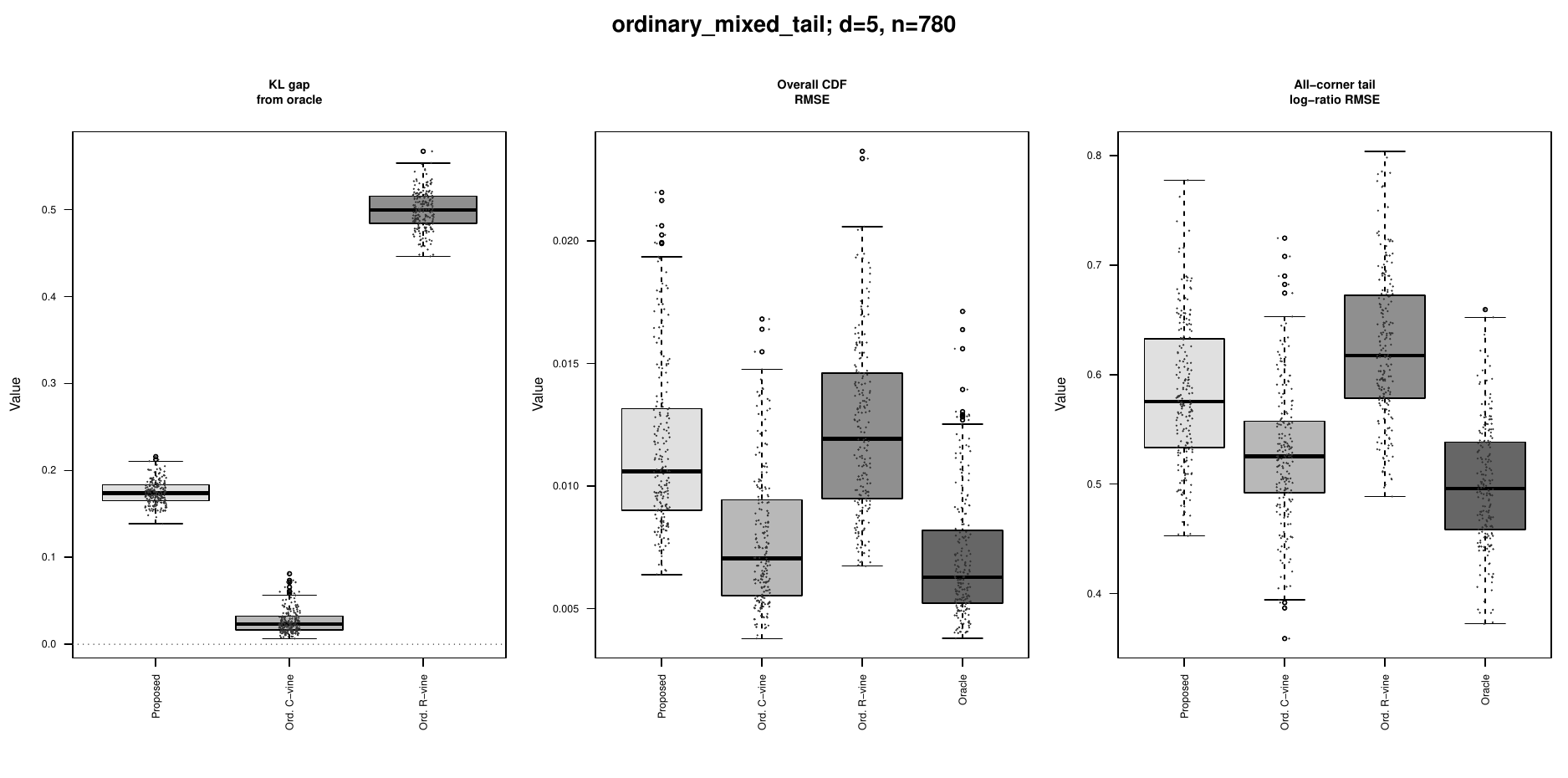}}
\caption{Three goodness-of-fit criteria for scenario ordinary\_mixed\_tail.}
\label{fig:ordinary_mixed_tail-n780}
\end{figure}

\section{Empirical study - medical expenditures}
\label{sec:empirical}
Understanding medical expenditures is crucial for evaluating healthcare costs, insurance planning, and policy decisions. In this empirical study, we analyze a dataset of Medical Expenditure Panel Survey (MEPS) \cite{AHRQ2026} to assess the dependence patterns of medical expenditures among frequent users of healthcare services. We use the 2023 and 2024 full year consolidated data for our analysis, and define frequent users as individuals with positive expenditures in both years in each of the following four categories:

\begin{center}
\begin{tabular}{ll}
\toprule
group & MEPS components \\
\midrule
\textsc{office}   & office-based visits (\texttt{OBV}) \\
\textsc{hospital} & outpatient $+$ emergency room $+$ inpatient
                    (\texttt{OPT}, \texttt{ERT}, \texttt{IPT}) \\
\textsc{rx}       & prescription medicines (\texttt{RX}) \\
\textsc{dental}   & dental visits (\texttt{DVT}) \\
\bottomrule
\end{tabular}
\end{center}

For each group, we use the total amount paid for that service by any source including out-of-pocket payments, private insurance, and public insurance programs.
Figure \ref{fig:pairs} shows the scatter plots of the eight variables on marginal normal scores, along with the median spending for each variable, Kendall's $\hat\tau$, and some crude tail estimates. There are noticeable differences among these dependence patterns: some pairs exhibit strong positive dependence, while others show weaker or negligible dependence; some pairs display asymmetry between the lower and upper tails, while others appear more symmetric; most pairs exhibit some degree of positive dependence, while dental spending appears relatively independent of the other categories and even a little negatively dependent on hospital spending.
Both temporal and cross-sectional dependence structures are evident in the scatter plots, with the four temporal panels being visibly different from the other 24; the strongest dependence is between \textsc{rx}\,'23 and \textsc{rx}\,'24, as indicated by the scatter plot in row 7, column 3.
These observations highlight the variability of dependence structures across different types of medical expenditures, emphasizing the need for flexible models that can capture both temporal and cross-sectional dependence, strong and/or weak dependence, as well as potential asymmetries in the tails.

\begin{figure}[htbp]
\centering
\includegraphics[width=\textwidth]{\detokenize{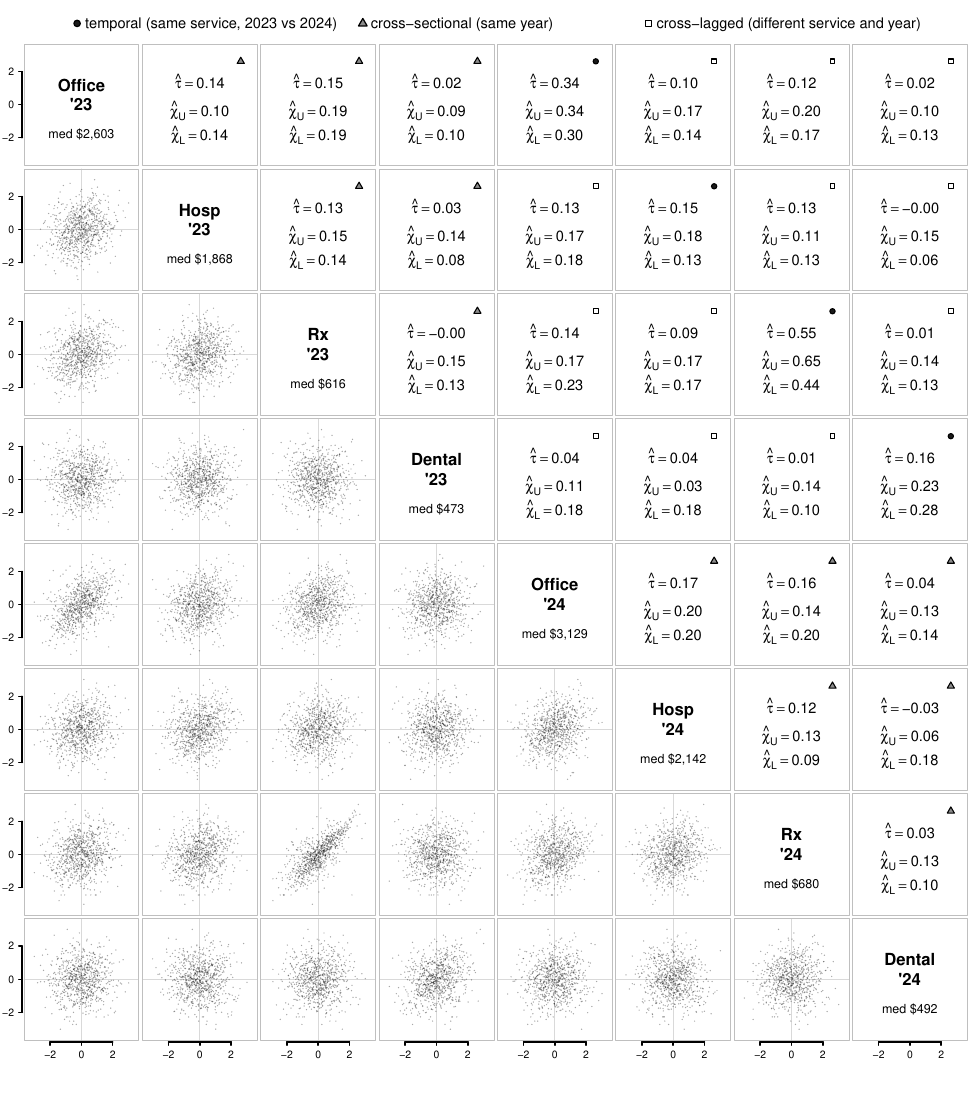}}
\caption{Scatter plots of the eight variables on marginal normal scores, as well as the median spending for each variable (med), estimated Kendall's $\hat\tau$, some crude tail estimates $\hat\chi_U = \widehat{\P}(U > 0.9, V > 0.9)/0.1$ and 
$\hat\chi_L = \widehat{\P}(U \le 0.1, V \le 0.1)/0.1$.}
\label{fig:pairs}
\end{figure}

We fit a C-vine with the following root order
\begin{align}
\textsc{rx'23} \rightarrow \textsc{office'23} \rightarrow
\textsc{hospital'23} \rightarrow \textsc{dental'23} \rightarrow  \notag \\
\textsc{rx'24} \rightarrow \textsc{office'24} \rightarrow
\textsc{hospital'24} \rightarrow \textsc{dental'24}, \notag
\end{align}
where the order reflects descending marginal dependence primarily and the chronological order secondarily.

There can be alternative root orders, but for this data the stronger dependence exists mainly in the temporal dependence within the same type of service over the two years, so temporal dependence conditional on those cross-sectional variables is not very different than that of the unconditional temporal dependence; see Table \ref{tab:temporal} for comparison of temporal dependence with and without conditioning on the cross-sectional variables. Therefore, the choice of the root order is not critical for capturing the main dependence patterns in the data.

There are 28 edges in total, and every two from the eight variables will be coupled once in the vine structure, either directly or conditionally. Each edge is modeled by the FRA1 copula or its $90^{\circ}$ rotation; no other copula family enters the model at any point. Estimation is performed by sequential maximum likelihood, meaning that parameters for each edge are estimated conditional on the previously estimated edges. The model has $56$ parameters and on average there are about 14 observations per parameter.

Figure~\ref{fig:map} plots the fitted $(\eta, \theta)$ of every edge. We use \rgm{+-} to indicate asymptotic dependence ($\theta>0$) in the upper tail and asymptotic independence ($\eta<0$) in the lower tail, and similarly for the other three quadrants: \rgm{-+}, \rgm{--}, and \rgm{++}. Only the seven tree-1 edges are unconditional, and the others are conditional pair copulas. The four temporal edges are labelled by service in the plot, while others are indicated using different symbols based on their type of dependence. For the four labelled temporal edges, only \textsc{Rx} is based on unconditional copula as the C-vine starts with \textsc{Rx'23}, and for all temporal dependence, only that between \textsc{Rx'23} and \textsc{Rx'24} is directly modeled in the root tree; the rest are the estimated parameters for the conditional copulas. However, the conditional copulas are quite close to their unconditional ones as the conditioning variables have much weaker dependence on the temporal edges and cannot absorb much of the temporal dependence; see Table~\ref{tab:temporal}.

From Figure~\ref{fig:map}, we find that there are no edges in the \rgm{-+} quadrant. The strongest dependence both in the lower and upper tails appears in the temporal dependence between \textsc{Rx} and between  \textsc{Office}, which may suggest that those who use prescription drugs regularly in one year are likely to do so in the next year as well, and similarly for office visits. However, hospital and dental visits exhibit weaker tail dependence overall, although dependence in the upper tail appears to be stronger than that in the lower tail. Most of the cross-sectional and cross-lagged edges fall into the \rgm{-}\rgm{-} quadrant, indicating weaker dependence in both tails. Only a few edges need a rotation of the copula to capture slightly negative dependence that may be associated with dental service usage.

\begin{figure}[H]
\centering
\includegraphics[width=\textwidth]{\detokenize{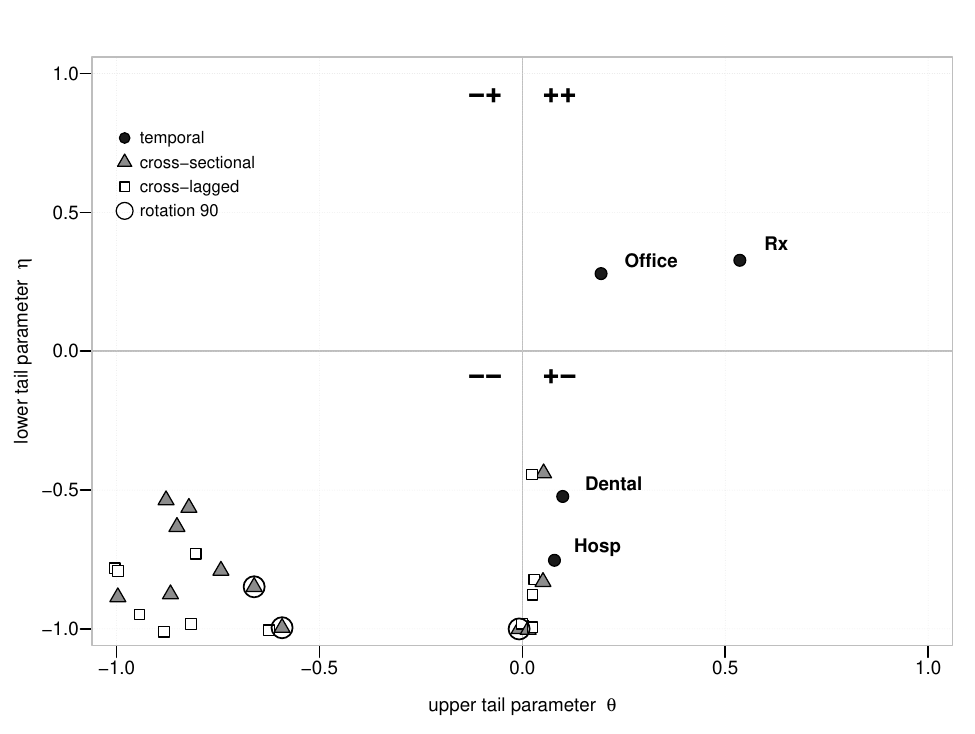}}
\caption{Fitted $(\eta,\theta)$ for all 28 edges of the FRA1 C-vine copula of the eight variables.}
\label{fig:map}
\end{figure}

Table~\ref{tab:temporal} shows the comparison of the four temporal edges, with and without conditioning on cross-sectional variables. The tail dependence regimes remain unchanged. The
largest change is \textsc{office}, whose $\lU$ goes from 0.261 conditional to
0.287 unconditional, and the reason is that, as a cross-sectional variable to \textsc{office}, \textsc{rx'23} has relatively strong dependence with both years' office spending, and conditioning on \textsc{rx'23} removes some of that dependence.

\begin{table}[H]
\centering
\caption{Estimates for the four temporal edges, with and without conditioning on cross-sectional variables.}
\label{tab:temporal}
\setlength{\tabcolsep}{5pt}
\begin{tabular}{cllrrcrrrr}
\toprule
tree & service & conditioned on & $\theta$ & $\eta$ & regime & $\lU$ & $\kU$ & $\lL$ & $\kL$ \\
\midrule
\multicolumn{10}{l}{\textit{As fitted in the vine}}\\
1 & \textsc{rx}       & ---                        &  0.546   &  0.333   & \rgm{++} & 0.630 & 1.00 & 0.250 & 1.00 \\
2 & \textsc{office}   & \textsc{rx'23}           &  0.202   &  0.271   & \rgm{++} & 0.261 & 1.00 & 0.155 & 1.00 \\
3 & \textsc{hospital} & $+$ \textsc{office'23}   &  0.078   & $-0.763$ & \rgm{+-} & 0.105 & 1.00 & 0     & 1.65 \\
4 & \textsc{dental}   & $+$ \textsc{hospital'23} &  0.104   & $-0.534$ & \rgm{+-} & 0.140 & 1.00 & 0     & 1.45 \\
\addlinespace
\multicolumn{10}{l}{\textit{Fitted unconditionally}}\\
  & \textsc{rx}       & ---                        &  0.546   &  0.333   & \rgm{++} & 0.630 & 1.00 & 0.250 & 1.00 \\
  & \textsc{office}   & ---                        &  0.223   &  0.300   & \rgm{++} & 0.287 & 1.00 & 0.198 & 1.00 \\
  & \textsc{hospital} & ---                        &  0.098   & $-0.728$ & \rgm{+-} & 0.132 & 1.00 & 0     & 1.66 \\
  & \textsc{dental}   & ---                        &  0.106   & $-0.543$ & \rgm{+-} & 0.141 & 1.00 & 0     & 1.46 \\
\bottomrule
\end{tabular}
\end{table}

With the fitted parameters in Figure~\ref{fig:map} and  Table~\ref{tab:temporal}, we can see that the FRA1 C-vine copula captures various dependence patterns over years and across different services effectively, and a single copula model is sufficient to describe the complex dependence structure in this data set and provides a direct comparison and interpretation of the dependence patterns.

\section{Concluding remarks}
\label{sec:conclude}
``All models are wrong, but some are useful.'' \citep{Box1976} Like any parametric copula model, the FRA1 copula is another candidate for capturing the complex dependence structure in multivariate data. What makes it appealing is its closed-forms and its ability to use only two parameters within a single model to capture both asymptotic dependence and asymptotic independence for both the lower and upper tails of a joint distribution. To the best of our knowledge, it is so far the only one that has a closed-form density function while having full-range tail dependence in both lower and upper tails. 

Because of these properties, many multivariate models can be built upon such a flexible copula and many sophisticated applications that demand a flexible dependence structure while requiring likelihood-based inference can benefit from it. With the simulation study, we use the FRA1 copula as the only building block for constructing C-vine copulas, while achieving comparable performance to the ordinary C-vine copulas constructed from selecting various parametric copulas. On the one hand, the FRA1 copula can serve as a candidate pair-copula to be selected by data for ordinary vine copulas, and on the other hand, the FRA1 copula itself can be used to construct a vine copula to avoid pair-copula model selection, be potentially more parsimonious, and more interpretable. With the empirical study, we demonstrate its practical usefulness in capturing various dependence patterns in temporal and cross-sectional dependence of a medical expenditure data set. Using only one flexible copula makes comparisons between different dependence patterns more straightforward and interpretable.

Limitations of the FRA1 copula include its nonregular likelihood function at $(0,0)$ and no perfect symmetry between the lower and upper tails. For the former, some optimized estimation techniques designed specifically for the FRA1 copula are required to further speed up the estimation process, while for the latter, potential extensions such as a mixture of the copula and its survival copula could be considered to achieve perfect reflection symmetry.

\appendix

\section{Proof of Proposition \ref{prop:generator-properties}}
\label{appendix:generator-properties}

To prove Proposition \ref{prop:generator-properties}, we first prove some lemmas, where properties of Bernstein functions, complete Bernstein functions, and Stieltjes functions play an important role in the proofs. We first give the definitions of these functions, and then list some known results taken from \cite{Schilling2012} that will be used throughout the proofs.

\begin{enumerate}
\item[(S1)] A function $f:(0,\oo)\to[0,\oo)$ belongs to $\St$, the set of \emph{Stieltjes functions}, if and only if it can be represented as $f(s)= a/s + b+\int_{(0,\oo)}(s+r)^{-1}\,\mu(d r)$ with $a,b\geq0$ and $\int_{(0,\oo)}(1+r)^{-1}\mu(d r)<\oo$ (Definition~2.1 of \cite{Schilling2012}).  In particular the constant function $1$ belongs to $\St$ ($a=0$, $b=1$, $\mu=0$).
\item[(S2)] $\St\subset\CM$: indeed $s\mapsto(s+r)^{-1} =\int_0^\oo e^{-sx}e^{-rx}d x$ is completely monotone for each $r\geq0$, and $\CM$ is a convex cone closed under positive mixtures (Corollary~1.6 of \cite{Schilling2012}).
\item[(S3)] A function $f:(0,\oo)\to [0, \oo)$ is a \emph{Bernstein function}, denoted as $f\in\BF$, if and only if $f\in C^\oo((0,\oo))$ and $f'\in\CM$ (Definition~3.1 of \cite{Schilling2012}).
\item[(S4)] A function $f:(0,\oo)\to [0, \oo)$ is a \emph{complete Bernstein function}, denoted as $f \in \CBF$, if it can be characterized by $f(t) =a+bt+ \int_{(0,\oo)}t/(t+r)\,\rho(d r), a,b\geq0$, with $\int_{(0,\oo)} (1+r)^{-1} \,\rho(d r) <\oo$ (Definition~6.1 of \cite{Schilling2012}).
\item[(S5)] For $f\not\equiv0$: $f\in\St$ if and only if $1/f\in\CBF$ (Theorem~7.3 of \cite{Schilling2012}).
\item[(S6)] $\CBF$ is a convex cone containing the nonnegative constants and the identity $\mathrm{id}:s\mapsto s$; it is closed under positive mixtures, under pointwise limits (Corollary~7.6 of \cite{Schilling2012}) and under composition (Corollary~7.9 of \cite{Schilling2012}).
\item[(S7)] Every $f\in\BF\supset\CBF$ is nonnegative, nondecreasing, concave and $C^\oo$ on $(0,\oo)$. (Definitions~3.1 and 6.1 of \cite{Schilling2012})
\item[(S8)] If $g\in\CM$ and $f\in\BF$, then $g\circ f\in\CM$ (Theorem~3.7 of \cite{Schilling2012}).
\end{enumerate}

\bigskip

\begin{lemma}
\label{lem:F-prime-Stieltjes}
For any $0\leq p\leq1$, $F_p' \in \St$, $F_p\in\BF$, and $q_b \in \CBF$ with 
\begin{align}
q_p(s):=\frac{1}{F_p'(s)}, \quad s > 0. \notag
\end{align}
\end{lemma}
 
\pf The condition $A(s)=\arcosh(1+s)$ implies that for $s>0$
\begin{align}
\cosh(A(s))=1+s,\quad \sinh(A(s))=\sqrt{\cosh^2(A(s))-1}=\sqrt{s(s+2)}=:\Delta(s)>0. \label{eq:cosh-sinh-A}
\end{align}
Therefore,
$A'(s)\sinh(A(s))=1$, and thus,
\begin{align}
A'(s)=\frac{1}{\sinh(A(s))}=\frac{1}{\Delta(s)}, \notag
\end{align}
together with the chain rule, implies that
\begin{align}
F_p'(s)=\frac1{p^2}\,p\sinh(pA(s))\,A'(s)  = \frac{\sinh(pA(s))}{p\,\sinh(A(s))}, \quad 0\leq p\leq 1, \notag
\end{align}
that is, $F_p$ is $C^\oo$ on $(0,\oo)$.
With $\lim_{p\downarrow0} \sinh(p\beta)/p=\beta$ and $\lim_{p\downarrow0} \sin(\pi p)/(\pi p)=1$, we have $F_0'(s)=A(s)/\Delta(s)$. Since $A(s)>0$,
\begin{align}
F_p'(s)>0, \quad \forall s>0,\ 0\leq p\leq 1,\label{eq:F-prime-positive}
\end{align}
and clearly, $F_p\geq0$ with $F_p(0+)=0$. For $p=1$, Eq.~\eqref{eq:cosh-sinh-A} leads to $F_1(s)=\cosh(A(s))-1=s$ and $F_1'\equiv1$.
 
Taking $b=p, c=1/2$, and $t=\beta/2$ in Formula 2.4.6.27 of \cite{Prudnikov1986} implies that for any $\beta>0$ and $0 < p<1$,
\begin{align}
\frac{\sinh(p\beta)}{p\sinh(\beta)} =\frac{\sin(\pi p)}{\pi p} \int_0^\oo\frac{\cosh(pw)}{\cosh(\beta)+\cosh(w)}\, d w. \label{eq:key-integral}
\end{align}
The above equality for the limiting case $p \ra 0$ can be obtained similarly. Hence, Eq.~\eqref{eq:key-integral} holds for all $0 \leq p < 1$. Then, we fix $s>0$ and take $\beta=A(s)$ in Eq.~\eqref{eq:key-integral}, so that $\cosh(\beta)=1+s$ and $\sinh\beta=\Delta(s)$ by Eq.~\eqref{eq:cosh-sinh-A}.  Letting $r=1+\cosh(w)$, a bijection from $(0,\oo)$ onto $(2,\oo)$. Then changes of variables give that, for any $0\leq p<1$ and $s>0$,
\begin{align}
F_p'(s) =\int_{(0,\oo)}\frac{1}{s+r}\,\mu_p(d r), \quad \mu_p(d r)=\frac{\sin(\pi p)}{\pi p}\, \frac{\cosh(pA(r-2))}{\Delta(r-2)}\,\mathbf \mathds{1}_{(2,\oo)}(r) d r, \label{eq:F-prime-Stieltjes-form}
\end{align}
where $\mu_p$ is a positive measure on $(0,\oo)$. Taking $s=1$ in Eq.~\eqref{eq:F-prime-Stieltjes-form},
\begin{align}
\int_{(0,\oo)}\frac{\mu_p(d r)}{1+r}=F_p'(1)<\oo. \notag
\end{align}
That is, (S1) holds, and therefore, $F_p'\in\St$ (with $a=b=0$) for every $0\leq p<1$.
As $F_1'\equiv1$, which trivially belongs to $\St$, we have
\begin{align}
F_p'\in\St\quad\text{for every }0\leq p\leq 1. \notag
\end{align}
By (S2), $F_p'\in\CM$. Also, $F_p \geq 0$ and $F_p\in C^\oo((0,\oo))$, so (S3) implies that
$F_p\in\BF$.

Finally, $F_p'>0$ on $(0,\oo)$ by Eq.~\eqref{eq:F-prime-positive}, so $F_p'$ is a nonzero Stieltjes function, and then (S5) completes the proof. \QED

\begin{lemma}
\label{lem:flow}
Suppose that for any $x >0$, the map $u \mapsto \Theta_u(x) \in (0, \oo)$ is differentiable on $[0,z]$, $z >0$, and 
\begin{align}
\frac{\partial}{\partial u}\Theta_u(x)=q(\Theta_u(x)), \quad \Theta_0(x)=x,\label{eq:flow-ode}
\end{align}
where $q\in\CBF$. Then $\Theta_z\in\CBF$.
\end{lemma}

\pf
By (S7), $q \in C^\oo((0,\oo))$, nonnegative, and nondecreasing. For a given $x>0$, let $W(u):=\Theta_u(x)$, then $q(W(u))$ is continuous and nonnegative, so $W$ is differentiable and nondecreasing, and
\begin{align}
0<x=W(0)\leq W(u)\leq W(z)=\Theta_z(x)=:M_x, \quad u\in[0,z]. \notag
\end{align}
Therefore, on the compact interval $I_x=[x,M_x]\subset(0,\oo)$, the following functions
\begin{align}
Q_x:=\max_{I_x}q, \quad \Lambda_x:=\max_{I_x}q' \notag
\end{align}
are finite.

\medskip
Let $n\in\mathbb N$, $h=z/n$, and 
\begin{align}
\mathcal E_h=\mathrm{id}+h\,q, \quad E_n=\underbrace{\mathcal E_h\circ\cdots\circ\mathcal E_h}_{n\ \text{times}}, \notag
\end{align}
where $\mathrm{id}$ is the identity map. By (S6) that $\CBF$ is a convex cone containing $\mathrm{id}$, we have $\mathcal E_h\in\CBF$. Also, $\mathcal E_h(y)\geq y>0$, so $\mathcal E_h$ maps $(0,\oo)$ into itself and is well defined. Because $\CBF$ is closed under composition (S6), $E_n\in\CBF$.

\medskip

Let $y_0=x$ and $y_{k+1}=\mathcal E_h(y_k)$, so that $E_n(x)=y_n$. Then $y_k\geq x$ because $q\geq0$. By induction, $y_k\leq W(kh)$: the case $k=0$ is trivial; assume that $y_k\leq W(kh)$, and because $W$, $q$, and $\mathcal E_h$ are nondecreasing, we have
\begin{align}
W((k+1)h)=W(kh)+\int_{kh}^{(k+1)h}q(W(u)) d u \ \geq\ \mathcal E_h(W(kh)) \ \geq\ \mathcal E_h(y_k)=y_{k+1}. \notag
\end{align}
Therefore, $x\leq y_k\leq W(kh)\leq M_x$ for $0\leq k\leq n$.

\medskip

Now, we are ready to prove that $E_n(x)\to\Theta_z(x)$ as $n\to\oo$. Let 
\begin{align}
e_k:=W(kh)-y_k, \notag
\end{align}
so $e_k\geq0$ and $e_0=0$ by the argument above. Subtracting both sides of $y_{k+1}=y_k+hq(y_k)$ from $W((k+1)h)=W(kh)+\int_{kh}^{(k+1)h}q(W(u)) d u$, we have
\begin{align}
e_{k+1}=e_k+\int_{kh}^{(k+1)h}\left[\,q(W(u))-q(y_k)\,\right] d u .
\label{eq:error-recursion}
\end{align}
For $u\in[kh,(k+1)h]$ both $y_k$ and $W(u)$ lie in $I_x$, and $y_k\leq W(kh)\leq W(u)$. Also, $W(u)-W(kh)\leq hQ_x$, so by the mean value theorem, $0\leq q(W(u))-q(y_k)\leq\Lambda_x\bigl(W(u)-y_k\bigr) \leq\Lambda_x(e_k+hQ_x)$. Then Eq.~\eqref{eq:error-recursion} leads to $e_{k+1}\leq(1+h\Lambda_x)e_k +\Lambda_xQ_xh^{2}$.

Iterating from $e_0=0$, for a finite $z>0$, when $n$ is sufficiently large, we have
\begin{align}
0 \leq \Theta_z(x)-E_n(x)=e_n \leq\ Q_xh\left((1+h\Lambda_x)^{n}-1\right) \leq\ \frac{z\,Q_x\left(e^{\Lambda_x z}-1\right)}{n}, \notag
\end{align}
where the last inequality is due to $(1+z)^n\leq e^{nz}$ for $z\geq0$. Therefore, $E_n\to\Theta_z$ pointwise on $(0,\oo)$ as $n\to\oo$, with finite limit. Together with the fact that $E_n\in\CBF$ we just proved, we conclude that $\Theta_z\in\CBF$ because $\CBF$ is closed under pointwise limits (S6). \QED
 
\begin{lemma}
\label{lem:operations}
Let $f\in\CBF$, $f\not\equiv0$.  Then
\begin{enumerate}
\item[(a)] $f-f(0+)\in\CBF$;
\item[(b)] $s\mapsto f(1/s)$ belongs to $\St$, and therefore, $f^{\ast}(s)=1/f(1/s)$ belongs to $\CBF$.
\end{enumerate}
\end{lemma}
 
\pf (a) Using the representation in (S4), we have $a=f(0+)$, so $f-f(0+)=bs+\int\frac{s}{s+t}\sigma(d t)$ again satisfies the representation with $a=0$, which proves the claim.

\medskip

(b) Replacing all $s$ by $1/s$ in the representation in (S4),
\begin{align}
f(1/s)=\frac{b}{s}+a+\int_{(0,\oo)}\frac{1}{s+1/t}\,\frac{\sigma(d t)}{t} =\frac{b}{s}+a+\int_{(0,\oo)}\frac{1}{s+r}\,\mu(d r), \notag
\end{align}
where $\mu$ is the image measure of $\sigma$ under the map $t\mapsto1/t$. Because
$\int_{(0,\oo)}(1+r)^{-1}\mu(d r) =\int_{(0,\oo)}\frac{1}{1+1/t}\frac{\sigma(d t)}{t} =\int_{(0,\oo)}\frac{\sigma(d t)}{1+t}<\oo$, $f(1/\cdot)\in\St$. Also, $f \not \equiv 0$, so $f^\ast\in\CBF$ by (S5). \QED
 
\begin{lemma}
\label{lem:T-S-J}
For every $0\leq p\leq1$, the functions $T_p$, $S_p$ and $J_p$ all belong to $\CBF$.
\end{lemma}

\pf By Lemma~\ref{lem:F-prime-Stieltjes}, $F_p\in\BF$ with $F_p'>0$ on $(0,\oo)$. Moreover, $F_p(0)=0$ and $F_p(s)\to\oo$ as $s\to\oo$ because $A(s)\to\oo$. Hence, $F_p$ a strictly increasing bijection of $[0,\oo)$ onto itself, $F_p^{-1} \in C^\oo((0,\oo))$, and
\begin{align}
(F_p^{-1})'(y)=\frac{1}{F_p'(F_p^{-1}(y))}=q_p(F_p^{-1}(y)). \label{eq:inverse-derivative}
\end{align}
Therefore, $T_p=F_p^{-1}(F_p+1)$ is well defined. Clearly, $d_p=F_p^{-1}(1)\in(0,\oo)$, and $S_p=T_p-d_p$ is also well defined. Since $F_p(T_p(x))=F_p(x)+1>F_p(x)$ and $F_p$ is increasing, $T_p(x)>x$. Also,
\begin{align}
S_p(0)=0,\quad S_p(x)>x-d_p\to\oo \mbox{ as } x\to\oo, \quad S_p>0\ \text{on}\ (0,\oo), \label{eq:S-properties}
\end{align}
so $J_p(t)=(S_p(1/t))^{-1}$ is well defined for $t>0$, with $J_p(0+)=0$ and $J_p(\oo)=\oo$.

\medskip

For $z \geq 0$ and $x>0$,  let $\Theta_z(x)=F_p^{-1}(F_p(x) + z)$, which is well defined because $F_p$ is onto $[0,\oo)$.  By Eq.~\eqref{eq:inverse-derivative},
\begin{align}
\frac{\partial}{\partial z}\Theta_z(x) = (F_p^{-1})'(F_p(x)+z) = q_p\left(F_p^{-1}(F_p(x)+z)\right)=q_p(\Theta_z(x)), \quad \Theta_0(x)=x, \notag
\end{align}
so $z \mapsto\Theta_z(x)$ solves Eq.~\eqref{eq:flow-ode}, and $T_p=\Theta_1(x)$. 
By Lemma~\ref{lem:F-prime-Stieltjes}, $q_p\in\CBF$ for any $0\leq p\leq1$. By Lemma~\ref{lem:flow} with $q=q_p$ and $z=1$,
\begin{align}
T_p\in\CBF ,\quad 0\leq p\leq1 .
\notag
\end{align}
Moreover, by Lemma~\ref{lem:operations}(a),
\begin{align}
S_p=T_p-T_p(0+)=T_p-d_p\in\CBF ,\quad S_p(0+)=0. \notag
\end{align}
Because $S_p \not \equiv 0$ by Eq.~\eqref{eq:S-properties}, Lemma~\ref{lem:operations}(b) implies that
\begin{align}
J_p=\frac{1}{S_p(1/\cdot)}\in\CBF ,\quad 0\leq p\leq 1, \notag
\end{align}
which completes the proof. \QED
 
\begin{lemma}
\label{lem:K-a-W-theta}
For $0<a\leq1$, $K_a\in\CBF$, and for $-1\leq\theta<1$, $W_\theta\in\CBF$.
\end{lemma}
\pf Let $\Theta_z(t)=\left(t^{a}+z\right)^{1/a}$, then
\begin{align}
\frac{\partial}{\partial z} \Theta_z(t) = \frac{\partial}{\partial z}\left(t^{a}+z\right)^{1/a} = \frac1a\left(t^{a}+z\right)^{(1-a)/a} =\frac1a\Theta_z(t)^{1-a} =: q(\Theta_z(t)), \notag
\end{align}
with $q(y)=a^{-1}y^{1-a}$. Because $y\mapsto y^{\alpha} \in \CBF$ for $0 < \alpha < 1$ (\cite{Schilling2012}), and a constant $a^{-1} \in \CBF$, $q \in \CBF$ for $0<a\leq1$. Also, $\Theta_0(t)=t$. Hence, $\Theta_1\in\CBF$ by Lemma~\ref{lem:flow}.
Since $\Theta_1(0+)=1$ and $K_a(t)=\left(1+t^{a}\right)^{1/a}-1=\Theta_1(t)-\Theta_1(0+)$, Lemma~\ref{lem:operations}(a) implies that
\begin{align}
K_a \in \CBF, \quad 0<a\leq1, \label{eq:K-a-CBF}
\end{align}
with $K_a(0)=0$ and $K_a(\oo)=\oo$.

\medskip

Because $\CBF$ is closed under composition (S6), Lemma~\ref{lem:T-S-J} and Eq.~\eqref{eq:K-a-CBF} imply that
\begin{align}
\CBF \ni W_\theta=
\begin{cases}
J_{-\theta}, & -1\leq\theta\leq0,\\[1mm]
J_0\circ K_{1-\theta}, & 0<\theta<1,
\end{cases} \notag
\end{align}
Moreover, in the proof of Lemma~\ref{lem:T-S-J} we have shown that $J_p(0+)=0$ and $J_p(\oo)=\oo$ for every $0\leq p\leq1$, and $K_a(0)=0$ and $K_a(\oo)=\oo$ for every $0<a\leq1$. Therefore, as $t\to\oo$,
\begin{align}
W_\theta(0)=0, \quad W_\theta(t)\to\oo, \quad -1\leq\theta<1. \notag
\end{align} \QED
 
\begin{lemma}
\label{lem:G-properties}
For $-1\leq\eta<1$, $G_\eta\in\CM$ with $G_\eta(0)=1$ and $G_\eta(\oo)=0$.
\end{lemma}

\pf For $p\in[0,1]$, Lemma~\ref{lem:F-prime-Stieltjes} implies that $F_p\in\BF$. 
If $-1\leq\eta\leq0$, then $p=-\eta$, $g(u)=e^{-u}$, and $g \in \CM$ because $(-1)^ng^{(n)}(u)=e^{-u}\geq0$. Therefore, $G_\eta=g\circ F_p \in \CM$ by (S8). If $0<\eta<1$, then $p=\eta$ and $g(u)=(1+u/b)^{-b}$ with $b=b_\eta=(1-\eta)\eta^{-2}\in(0,\oo)$. Here, $g$ is the Laplace transform of a Gamma random variable, and therefore, $g\in\CM$, and $G_\eta=g\circ F_p \in \CM$ by (S8). 

Furthermore, $g(0)=1$ and $g(\oo)=0$ in both cases, and notice that $F_p(0)=0$ and $F_p(\oo)=\oo$, so that
$G_\eta(0)=1, G_\eta(\oo)=0$. \QED

\bigskip

\noindent \emph{Proof of Proposition \ref{prop:generator-properties}:} Lemmas \ref{lem:G-properties}, \ref{lem:K-a-W-theta}, and (S8) imply that
\begin{align}
\psi_{\eta,\theta}=G_\eta\circ W_\theta\in\CM. \notag
\end{align}
Moreover, 
\begin{align}
\psi_{\eta,\theta}(0)=G_\eta(W_\theta(0))=G_\eta(0)=1,  \quad \lim_{t\to\oo}\psi_{\eta,\theta}(t) =\lim_{s\to\oo}G_\eta(s)=0, \notag
\end{align}
the second limit because $W_\theta(t)\to\oo$ and $G_\eta$ is monotone. Both $G_\eta$ and $W_\theta$ are continuous on $[0,\oo)$, hence so is $\psi_{\eta,\theta}$. Finally, a completely monotone function is of the form $\psi(t)=\int_{[0,\oo)}e^{-tw}\sigma(d w)$ with $\sigma\geq0$; here $\sigma$ is a probability measure because $\psi_{\eta,\theta}(0)=1$, and $\sigma(\{0\}) =\lim_{t\to\oo}\psi_{\eta,\theta}(t)=0$, so
\begin{align}
\psi_{\eta,\theta}'(t)=-\int_{(0,\oo)}we^{-tw}\sigma(d w)<0, \quad t>0, \notag
\end{align}
i.e.\ $\psi_{\eta,\theta}$ is strictly decreasing.  It is therefore a continuous strictly decreasing bijection from $[0,\oo)$ onto $(0,1]$, so $\phi_{\eta,\theta}=\psi_{\eta,\theta}^{-1}$ is well defined on $(0,1]$ (with $\phi_{\eta,\theta}(0+)=\oo$).

With $\psi_{\eta,\theta}(0)=1$, $\psi_{\eta,\theta}(\oo)=0$, and $\psi_{\eta,\theta}$ continuous and strictly decreasing on $[0,\oo)$, it follows that $\phi_{\eta,\theta}=\psi_{\eta,\theta}^{-1}$ is well defined on $(0,1]$ and $C_{\eta,\theta}$ of Eq.~\eqref{eq:archimedean-copula} is a valid Archimedean copula; indeed $\psi_{\eta,\theta}$ generates a valid $d$-variate Archimedean copula for every $d\geq2$. \QED

\section{Details of the closed-form expressions}
\label{appendix:closed-form-expressions}
In addition to the closed-form expressions for quantities such as $F_p(s)$, $A(s)$, and $K_a(t)$ in Proposition \ref{prop:generator-properties}, we can derive following closed-form expressions:
\begin{align}
F_p^{-1}(y) &= \begin{cases}
\cosh\left(\frac{1}{p}\arcosh(1+p^2y)\right)-1, &0<p\leq1,\\[3mm]
\cosh(\sqrt{2y})-1, &p=0.
\end{cases} \notag \\
d_p &=
\begin{cases}
\cosh\left( \frac{1}{p}\arcosh(1+p^2) \right)-1, & 0<p\leq1,\\[3mm]
\cosh(\sqrt{2})-1, &p=0.
\end{cases} \notag \\
S_p^{-1}(y) &= F_p^{-1}\left(F_p(y+d_p)-1\right). \notag \\
J_p^{-1}(y) &=\left[ F_p^{-1} \left( F_p\left(d_p+\frac{1}{y}\right)-1 \right) \right]^{-1}, \qquad y>0. \notag \\
K_a^{-1}(y) & =[(1+y)^a-1]^{1/a}. \notag \\
W_\theta^{-1}(y) & = \begin{cases}
J_{-\theta}^{-1}(y), &-1\leq\theta\leq0,\\[2mm]
K_{1-\theta}^{-1}(J_0^{-1}(y)), &0<\theta<1.
\end{cases} \notag \\
G_\eta^{-1}(u) &= \begin{cases}
F_{-\eta}^{-1}(-\log u), &-1\leq\eta\leq0,\\[2mm]
F_\eta^{-1} \left[ b_\eta(u^{-1/b_\eta}-1) \right], &0<\eta<1.
\end{cases}\notag \\
\phi_{\eta,\theta}(u) & =\psi_{\eta,\theta}^{-1}(u) =
\begin{cases}
J_{-\theta}^{-1}\left(z_\eta(u)\right), &-1\leq\theta\leq0,\\[2mm]
K_{1-\theta}^{-1} \left[ J_0^{-1}(z_\eta(u)) \right], &0<\theta<1;
\end{cases} \qquad z_\eta(u) :=G_\eta^{-1}(u). \label{eq:psi-inverse}
\end{align}

Letting $\Delta(s) =\sqrt{s(s+2)}$, for $0<p\leq1$,
\begin{align}
F_p'(s) =\frac{\sinh(pA(s))}{p \Delta(s)}, \quad F_p''(s) = \frac{p\cosh(pA(s))\Delta(s) -(1+s)\sinh(pA(s))}{p\Delta(s)^3}, \notag
\end{align}
where $\Delta'(s)=(1+s)/\Delta(s)$. For $p=0$,
\begin{align}
F_0'(s)  = \frac{A(s)}{\Delta (s)}, \quad F_0''(s) = \frac{\Delta (s)-(1+s)A(s)}{\Delta(s)^3}, \notag
\end{align}
with $F_p'(0+) =1, F_p''(0+) =-(1-p^2)/3$.

Write $T=T_p(x)$, and then
\begin{align}
T_p'(x) = \frac{F_p'(x)}{F_p'(T)} \quad T_p''(x) =\frac{F_p''(x)}{F_p'(T)} -\frac{(F_p'(x))^2F_p''(T)}{(F_p'(T))^3}. \notag
\end{align}
Since $S_p=T_p-d_p$ and $d_p$ is constant with respect to $x$,
\begin{align}
S_p'(x)  &=T_p'(x), & S_p''(x) &=T_p''(x). \notag
\end{align}

For $t>0$, let $x=1/t$ and $S=S_p(x)$. Then,
\begin{align}
J_p'(t) = \frac{S_p'(x)}{t^2S^2}, \quad J_p''(t) = -\frac{2S_p'(x)}{t^3S^2} -\frac{S_p''(x)}{t^4S^2} +\frac{2(S_p'(x))^2}{t^4S^3}. \notag
\end{align}
Clearly,
\begin{align}
K_a'(t) =t^{a-1}(1+t^a)^{1/a-1}, \quad K_a''(t) = (a-1)t^{a-2}(1+t^a)^{1/a-2}. \notag
\end{align}
Therefore,
\begin{align}
W_\theta'(t) &=J_p'(t), \quad W_\theta''(t) =J_p''(t), \quad -1\leq\theta\leq0. \notag \\
W_\theta'(t) &=J_0'(k)K_a'(t), \quad W_\theta''(t) =J_0''(k)(K_a'(t))^2 +J_0'(k)K_a''(t), \quad a=1-\theta, \quad 0<\theta<1. \notag
\end{align}

For $-1\leq\eta\leq0$, let $p=-\eta$. Then,
\begin{align}
G_\eta'(s) =-F_p'(s)G_\eta(s), \quad G_\eta''(s) =\left[(F_p'(s))^2-F_p''(s)\right]G_\eta(s). \notag
\end{align}

For $0<\eta<1$, let $p=\eta$ and $Q_\eta(s) =1+F_p(s)/b_\eta$. Then,
\begin{align}
G_\eta'(s) &=-F_p'(s)Q_\eta(s)^{-b_\eta-1}, \notag \\
G_\eta''(s) &=-F_p''(s)Q_\eta(s)^{-b_\eta-1} +\frac{b_\eta+1}{b_\eta}(F_p'(s))^2Q_\eta(s)^{-b_\eta-2}. \notag
\end{align}

Write $w=W_\theta(t), w_1 =W_\theta'(t), w_2=W_\theta''(t)$, and then
\begin{align}
\psi_{\eta,\theta}'(t) =G_\eta'(w)w_1, \quad \psi_{\eta,\theta}''(t) =G_\eta''(w)w_1^2+G_\eta'(w)w_2. \label{eq:psi-derivatives}
\end{align}

\medskip

For $u,v\in(0,1)$, let $x =\phi_{\eta,\theta}(u), y =\phi_{\eta,\theta}(v)$, and then the copula density is
\begin{align}
c_{\eta,\theta}(u,v) = \frac{\psi_{\eta,\theta}''(x+y)}{\psi_{\eta,\theta}'(x) \psi_{\eta,\theta}'(y)}. \notag
\end{align}

Let $(U,V)$ be a random vector with copula $C_{\eta,\theta}$ defined in Eq.~\eqref{eq:archimedean-copula}. Then,
\begin{align}
h_{1\mid 2}(u\mid v) & :=F_{U\mid V=v}(u) =\frac{\partial}{\partial v}C_{\eta,\theta}(u,v) =\frac{\partial}{\partial v}\psi_{\eta,\theta}(\phi_{\eta,\theta}(u)+\phi_{\eta,\theta}(v)) \notag \\
& =\psi_{\eta,\theta}'(\phi_{\eta,\theta}(u)+\phi_{\eta,\theta}(v))\phi_{\eta,\theta}'(v) = \frac{\psi_{\eta,\theta}'(x+y)}{\psi_{\eta,\theta}'(y)}. \notag
\end{align}
Therefore, for a given $q\in(0,1)$ and $v\in(0,1)$, the inverse of the $h$-function with respect to its first argument is given by
\begin{align}
h_{1\mid2}^{-1}(q\mid v) =\psi_{\eta,\theta} \left[(\psi_{\eta,\theta}')^{-1} (q\psi_{\eta,\theta}'(y)) -y \right]. \label{eq:inverse-h-formal}
\end{align}

\section{Proof of Proposition \ref{prop:tail-map}}
\label{appendix:dependence}

\pf To prove the tail dependence properties, we need to first study the tail behavior of $G_\eta(s)$ and $W_\theta(s)$, respectively. 

\medskip

\underline{Part I: Tail behavior of $G_\eta(s)$:}

\medskip

Based on the definition of $G_\eta(s)$, $F_\eta(s)$ and thus $T_\eta(s)$ play the central role, we first analyze these functions as $s \to \infty$ and $s \downarrow 0$, respectively.
In what follows, we use $x=A(s)=\arcosh(1+s)$, so that $s=A^{-1}(x)=\cosh (x-1)$, and $L=L(s)=e^{x}=1+s+s\sqrt{1+2/s}$ so that $A(s)=\log L(s)$.

\underline{For $F_p(s)$, as $s \to \oo$}, $\Delta(s):=s\sqrt{1+2/s}=s+1-\frac{1}{2s}+O(s^{-2})$ by the binomial series, and $A=\log L=\log(2s)+\log(1+s^{-1}+O(s^{-2}))=\log(2s)+s^{-1}+O(s^{-2})$. Write 
\begin{align}
L=2s[1+s^{-1}+O(s^{-2})], \label{eq:L-large-s}
\end{align}
then $L^{p}=(2s)^{p}(1+ps^{-1}+O(s^{-2}))$. Because $F_p(s)=[\cosh(pA(s))-1]p^{-2} =[L^{p}+L^{-p}-2]/(2p^{2})$,
\begin{align}
F_p(s) \sim \frac{(2s)^{p}}{2p^{2}}, \quad s \to \infty, \label{eq:F-infty}
\end{align}
and thus $F_p\in\RV_{p}$ at $\infty$ for $p>0$. If $p=0$, then $F_0=\tfrac12A^{2}$ and thus,
\begin{align}
F_0(s) \sim \tfrac12\log^{2}(2s), \quad s \to \infty, \notag
\end{align}
and therefore, $F_0$ is slowly varying at $\oo$.

\medskip

\underline{For $F_p(s)$, as $s\downarrow0$}, and thus, $x \downarrow 0$, 
\begin{align}
s=\cosh(x) -1=x^{2}/2+ x^4/24 + O(x^{6}). \label{eq:s-small-x}
\end{align}
So, we can write $x^{2}=2s+ z s^{2}+O(s^{3})$ for an unknown constant $z$. Then, $x^{4}=4s^{2}+O(s^{3})$. Plugging these into Eq.~\eqref{eq:s-small-x} gives $z = -1/3$. Therefore,
\begin{align}
F_p(s)=&\frac{\cosh(p x)-1}{p^{2}} =\frac{x^{2}}{2}+\frac{p^{2}x^{4}}{24}+O(x^{6}) =\left(s-\frac{s^{2}}{6}\right)+\frac{p^{2}\cdot4s^{2}}{24}+O(s^{3}), \notag \\
   & = s-\frac{1-p^2}{6}s^2+O(s^3), \quad 0 \leq p \leq 1, \quad s \downarrow 0,  \label{eq:F-zero-derivation}
\end{align}
where the case $p=0$ satisfies because $F_0(s) = A^2(s)/2 = x^2/2 = s - s^2/6 + O(s^3)$.

\medskip

\underline{Tail behavior of $T_p$}:

\medskip

For $a\ge0$ and $k>0$, let $\tau_a(y)=y+a, m_k(y)=k y,\sigma_a=A\circ\tau_a\circ A^{-1}$, and then they are maps of $[0,\infty)$ into itself. Define, for $p >0$, $\nu_p(x):=\frac1p\arcosh\left(\cosh(px)+p^{2}\right)$. Then,
\begin{align}
T_p &=F_p^{-1}\circ\tau_{1}\circ F_p \notag \\
&=\left(A^{-1}\circ m_{p^{-1}}\circ A\circ m_{p^{2}}\right) \circ\tau_{1}\circ \left(m_{p^{-2}}\circ A^{-1}\circ m_{p}\circ A\right) \notag \\
&=A^{-1}\circ m_{p^{-1}}\circ A \circ\left(m_{p^{2}}\circ\tau_{1}\circ m_{p^{-2}}\right) \circ A^{-1}\circ m_{p}\circ A \notag \\
&=A^{-1}\circ m_{p^{-1}} \circ\left(A\circ\tau_{p^{2}}\circ A^{-1}\right) \circ m_{p}\circ A
\notag \\
&=A^{-1}\circ \left(m_{p^{-1}}\circ\sigma_{p^{2}}\circ m_{p}\right) \circ A \notag \\
&=A^{-1}\circ\nu_p\circ A. \label{eq:T-nu}
\end{align}
For $p=0$, define $\nu_0(x)=\sqrt{x^{2}+2}$, and then Eq. (\ref{eq:T-nu}) still holds. 

For a given $0<p\le 1$, let $s\to\infty$, so that $x=A(s)\to\infty$.  Write $\epsilon_p(x)=\nu_p(x)-x$, and recall $A^{-1}(z)=\cosh (z) -1$, $\cosh (x)=1+s$, and $\sinh (x)=\tfrac12\left(L-L^{-1}\right)$.  Then
\begin{align}
T_p(s) &=A^{-1}\left(\nu_p\left(A(s)\right)\right) =\cosh\left(x+\epsilon_p(x)\right)-1 \notag \\[1mm]
&=\cosh (x)-1+\epsilon_p(x)\sinh(x) +O\!\left(\epsilon_p(x)^{2}\cosh (x)\right) \label{eq:a2}\\[1mm]
&=s+\epsilon_p(x)\sinh(x) +O\!\left(\epsilon_p(x)^{2}\cosh (x)\right) \notag \\[1mm]
&=s+2pe^{-px}\cdot\frac{L-L^{-1}}{2} + O\!\left(L^{1-2p}\right) \label{eq:a4}\\[1mm]
&=s+p\,L^{1-p}+O\!\left(L^{-1}\right)+O\!\left(L^{1-2p}\right) \notag\\[1mm]
&=s+p\,2^{1-p}s^{1-p}\left[1+O\!\left(s^{-1}\right)\right] + O\!\left(s^{-1}\right)+O\!\left(s^{1-2p}\right) \notag\\[1mm]
&=s+p\,2^{1-p}s^{1-p}+o\!\left(s^{1-p}\right),  \qquad 0 < p \leq 1, \qquad s \to \oo . \label{eq:a7}
\end{align}
Eq. \eqref{eq:a2} is obtained via Taylor's expansion of $\cosh$ at $x$. For Eq. \eqref{eq:a4}, note that
\begin{align}
p \nu_p(x) & = \arcosh(\cosh(px)+p^{2}) =  \arcosh(\tfrac12e^{px}\left[1+2p^{2}e^{-px}+e^{-2px}\right]) \notag \\
  & =  px+\log\left(1+2p^{2}e^{-px}+e^{-2px}\right) + O\!\left(e^{-2px}\right) \label{eq:nu-infty} \\
  &= px+2p^{2}e^{-px}+O\!\left(e^{-2px}\right), \notag
\end{align}
where Eq. \eqref{eq:nu-infty} is due to the expansion $\arcosh (z)=\log(2z)+O(z^{-2})$ as $z\to\infty$. Therefore, $\epsilon_p(x)=2pe^{-px}+O(e^{-2px})$. Also, $\sinh(x)=\tfrac12(L-L^{-1})$, so, $O(e^{-2px})\sinh(x)=O(L^{-2p}\cdot L)$, and furthermore, $\epsilon_p^{2}\cosh (x)=O(L^{-2p}\cdot L)$, so we can conclude Eq. \eqref{eq:a4}. The remaining steps are due to Eq. \eqref{eq:L-large-s}.

For $p=0$, as $x \to \oo$, $\epsilon_0(x)=\sqrt{x^{2}+2}-x = x(\sqrt{1 + 2x^{-2}}-1) = x\left(1 + x^{-2} - 1\right) + O(x^{-3}) = x^{-1}+O(x^{-3})$. Then,
\begin{align}
T_0(s) &=\cosh\left(x+\epsilon_0(x)\right)-1 =s+\frac{\sinh(x)}{x}\left[1+O\!\left(x^{-2}\right)\right] +O\!\left(\frac{\cosh(x)}{x^{2}}\right) \notag\\[1mm]
&=s+\frac{s}{\log(2s)}+O\!\left(\frac{s}{\log^{2}s}\right), \notag
\end{align}
because $\sinh(x)=\sqrt{s(s+2)}=s[1+O(s^{-1})]$, $\cosh (x)=1+s$, and $x=\log(2s)+O(s^{-1})$, as $s\to\infty$.

\medskip

As $s\downarrow0$, $T_p(s) \to F_p^{-1}(1)=d_p$.
Because $F_p^{-1}(y) = \cosh\left(\tfrac1p\arcosh(1+p^{2}y)\right)-1$, $x_p^{*}:=A(d_p)=A(F_p^{-1}(1))=A\left(\cosh\left(\tfrac1p\arcosh(1+p^{2})\right)-1\right) = \tfrac1p\arcosh(1+p^{2})$ for $p>0$, and $x_0^{*}:=A(d_0)=A(F_0^{-1}(1))=A\left(\cosh(\sqrt{2})-1\right)=\sqrt{2}$ for $p=0$. Also, because $A\circ T_p=\nu_p\circ A$ and $A(0)=0$, we have $\nu_p(0) = A(d_p) = x_p^{*}$. Now, differentiating both sides of $T_p=A^{-1}\circ\nu_p\circ A$, we have,
\begin{align}
T_p'(s) &=\left(A^{-1}\right)'\!\left(\nu_p(x)\right)\cdot\nu_p'(x)\cdot A'(s) =\sinh\left(\nu_p(x)\right) \cdot\frac{\sinh (px)}{\sinh\left(p\nu_p(x)\right)} \cdot\frac{1}{\sinh (x)} \notag \\[1mm]
       &=\frac{\sinh\left(\nu_p(x)\right)}{\sinh\left(p\nu_p(x)\right)}\cdot\frac{\sinh (px)}{\sinh (x)}\rightarrow \frac{\sinh (x_p^{*})}{\sinh\left(px_p^{*}\right)} \cdot p =\frac{\sinh (x_p^{*})}{\sqrt{2+p^{2}}} =: \frac{1}{\beta_p}, \qquad x \downarrow 0. \notag
\end{align}
Also, because $T_p'=F_p'/(F_p'\circ T_p)$ and $F_p'(0+)=1$, we have $T_p'(0+)=1/F_p'(d_p)$. Therefore,
\begin{align}
\beta_p=F_p'(d_p)=\sqrt{2+p^{2}}/\sinh(x_p^{*}) = \frac{\sqrt{2+p^{2}}}{\sinh\left(p^{-1}\arcosh(1+p^{2})\right)}. \notag
\end{align}
Since $T_p$ is $C^{\infty}$ at $0$, $S_p(s)=T_p(s)-T_p(0)=T_p'(0+)s[1+O(s)]=\beta_p^{-1}s[1+O(s)]$ as $s\downarrow0$. Here $\beta_0=\sqrt2/\sinh(\sqrt 2)$, and $\beta_p\in(0,1]$ is increasing in $p$.

\medskip

\underline{Part II: Tail behavior of $W_\theta(t)$:}

\medskip

\underline{As $t \to \oo$}, from the expansion for $S_p$ above, for any \underline{$0\leq p\leq1$},
\begin{align}
J_p(t) &=\frac{1}{S_p(1/t)} =\frac{1}{\beta_p^{-1}t^{-1}\left[1+O(t^{-1})\right]} =\beta_p t\left[1+O(t^{-1})\right], \quad t\to\infty. \label{eq:J-infty-proof}
\end{align}

Also, for $0<a\leq1$, the binomial expansion leads to
\begin{align}
K_a(t) &=t\left(1+t^{-a}\right)^{1/a}-1 =t\left[1+\frac{1}{a}t^{-a}-t^{-1}+O(t^{-2a})\right] =t\left[1+O(t^{-a})\right], \quad t\to\infty. \label{eq:K-infty-proof}
\end{align}

Therefore,
\begin{align}
W_\theta(t) =\beta_{\vert \theta\vert}t\left[1+O\left(t^{-\xi}\right)\right],
\qquad t\to\infty, \qquad \xi =
\begin{cases}
1, & -1\leq\theta\leq0,\\[1mm]
1-\theta, & 0<\theta<1,
\end{cases} \label{eq:W-infty}
\end{align}
where the case for $-1\leq\theta\leq0$ follows directly from Eq.~\eqref{eq:J-infty-proof}, and for $0<\theta<1$, let $a=1-\theta$. Then $K_a(t)\to\infty$, and Eqs.~\eqref{eq:J-infty-proof} and \eqref{eq:K-infty-proof} lead to $J_0(K_a(t))=\beta_0t[1+O(t^{-a})]$. Thus, $W_\theta(t)$ is asymptotically linear with the slope $\beta_{\vert \theta \vert}$ as $t \to \oo$ for every $-1\leq\theta<1$.

\medskip

\underline{As $t \downarrow 0$}, \underline{if $0<p<1$}, let $s=1/t$ and $c_p=p2^{1-p}$, and then Eq.~\eqref{eq:a7} implies that
\begin{align}
S_p(1/t) =T_p(1/t)-d_p =t^{-1}\left(1+c_p t^{p}-d_p t+o(t^{p})\right) =t^{-1}\left(1+c_p t^{p}+o(t^{p})\right), \quad t \downarrow 0, \notag
\end{align}
where the constant $d_p$ is absorbed into $o(t^p)$ as $p<1$, $t=o(t^p)$. Therefore, $(1+t)^{-1}=1-t+o(t)$, $t\downarrow 0$, implies that 
\begin{align}
J_p(t) &=\frac{1}{S_p(1/t)} =t\left(1-c_pt^{p}+o(t^{p})\right) =t-p2^{1-p}t^{1+p}+o\left(t^{1+p}\right), \qquad t\downarrow 0. \label{eq:J-zero-interior}
\end{align}

\underline{If $p=1$}, then clearly,
\begin{align}
F_1(s)=s, \quad T_1(s)=s+1,\quad d_1=1,\quad S_1(s)=s,\quad J_1(s)=s. \label{eq:J-one-exact}
\end{align}

\underline{If $p=0$}, write $r=A(x)$ and $\delta_r=\sqrt{r^2+2}-r$. Then, as $x\to\infty$, i.e., as $r\to\infty$ and $\delta_r\to 0$, by binomial and Taylor expansions, one has
\begin{align}
\delta_r &=r^{-1}-\frac{1}{2}r^{-3}+O(r^{-5}), \notag \\[1mm]
\sinh(\delta_r) &=r^{-1}-\frac{1}{3}r^{-3}+O(r^{-5}), \notag \\[1mm]
\cosh(\delta_r)-1 &=\frac{1}{2}r^{-2}+O(r^{-4}), \notag
\end{align}

Because 
\begin{align}
T_0(x) & = \cosh(r+\delta_r)-1 = \cosh(r) \cosh(\delta_r) + \sinh(r) \sinh(\delta_r) - 1 \notag \\[1mm]
       & = (x+1)\cosh( \delta_r) + \sqrt{x(x+2)} \sinh(\delta_r) -1, \notag
\end{align}
and $x^{-1}=o(r^{-m})$ for $m>0$ as 
\begin{align}
r=A(x)=\log(2x)+x^{-1}+O(x^{-2}), \label{eq:A-refined}
\end{align}
\begin{align}
\frac{S_0(x)}{x} &=\frac{T_0(x)-d_0}{x} =1+\frac{x+1}{x}\left[\cosh(\delta_r)-1\right] +\frac{\sqrt{x(x+2)}}{x}\sinh(\delta_r)-\frac{d_0}{x} \notag \\[1mm]
&=1+\frac{1}{r}+\frac{1}{2r^2}+O(r^{-3}). \label{eq:S-zero-refined}
\end{align}

Letting $s=1/x$, $r_s=A(1/s)$, and $\ell_s=\log(2/s)$, as $x\to\infty$, implies that $r_s=\ell_s+O(s)$. Also, as $s \downarrow 0$, $(1+s)^{-1} = 1 -s +s^2+O(s^3)$. Therefore, Eq.~\eqref{eq:S-zero-refined} implies that
\begin{align}
J_0(s) &=\frac{1}{S_0(1/s)} =s\left[1+\frac{1}{r_s}+\frac{1}{2r_s^2}+O(r_s^{-3})\right]^{-1} =s\left[1-\frac{1}{r_s}+\frac{1}{2r_s^2}+O(r_s^{-3})\right] \notag \\[1mm]
&=s\left[1-\frac{1}{\ell_s}+\frac{1}{2\ell_s^2}+O(\ell_s^{-3})\right],
\qquad s\downarrow0.
\label{eq:J-zero-boundary}
\end{align}
Finally, the binomial expansion leads to the following expression:
\begin{align}
K_a(s) &=\left(1+s^a\right)^{1/a}-1 =\frac{s^a}{a}\left[1+O(s^a)\right], \qquad s\downarrow0.\label{eq:K-zero-proof}
\end{align}
and thus $\log(2/K_a(s))=a\log(1/s)+O(1)$, as $s\downarrow0$.

\medskip

\underline{For $0<\theta<1$}, let $a=1-\theta$. Then, Eq.~\eqref{eq:K-zero-proof} and \eqref{eq:J-zero-boundary} imply that
\begin{align}
W_\theta(s) &=J_0\left(K_a(s)\right) =\frac{s^a}{a}\left[1+O\left(\frac{1}{\log(1/s)}\right)\right],
\qquad s\downarrow0.
\label{eq:W-positive-zero}
\end{align}

Combining Eqs.~\eqref{eq:J-zero-interior}, \eqref{eq:J-one-exact}, \eqref{eq:J-zero-boundary}, and \eqref{eq:W-positive-zero}, we have

\begin{align}
W_\theta(s) &=
\begin{cases}
s, & \theta=-1,\\[1mm]
s-(-\theta)2^{1+\theta}s^{1-\theta}+o(s^{1-\theta}), & -1<\theta<0,\\[1mm]
\displaystyle s\left[1-\frac{1}{\ell_s}+\frac{1}{2\ell_s^2}+O(\ell_s^{-3})\right], & \theta=0,\\[3mm]
\displaystyle \frac{s^{1-\theta}}{1-\theta} \left[1+O\left(\frac{1}{\log(1/s)}\right)\right], &0<\theta<1,
\end{cases}
\qquad s\downarrow0, \label{eq:W-zero-all}
\end{align}
where $\ell_s=\log(2/s)$.

\medskip

\underline{Part III: Lower tail:}

\medskip

Write $\psi=\psi_{\eta,\theta} = G_{\eta} \circ W_\theta$, $t=\psi^{-1}(u)$, so that $t\to\infty$ as $u\downarrow0$. Then
\begin{align}
C(u,u) =\psi\left(\psi^{-1}(u)+\psi^{-1}(u)\right) =\psi(2t), && \lambda_L =\lim_{t\to\infty}\frac{\psi(2t)}{\psi(t)}, && \kappa_L &=\lim_{t\to\infty}\frac{\log\psi(2t)}{\log\psi(t)}, \notag
\end{align}
whenever the limits exist.

By Eq.~\eqref{eq:W-infty}, $W_\theta(t) =\beta_{\vert \theta\vert}t\left[1+O\left(t^{-\xi}\right)\right]$. \underline{For $p>0$}, Eqs.~\eqref{eq:F-infty} and \eqref{eq:W-infty} implies that
\begin{align}
H_p(t) :=F_p\left(W_\theta(t)\right) \sim\frac{\left(2W_\theta(t)\right)^{p}}{2p^2} \sim\frac{(2\beta_{\vert \theta\vert})^p}{2p^2}t^p, \qquad t\to\infty. \label{eq:Hpt-infty-proof}
\end{align}
Therefore, $H_p\in\RV_p$ at infinity, and
\begin{align}
\frac{H_p(2t)}{H_p(t)} &\rightarrow 2^p, \qquad t\to\infty. \notag
\end{align}
The above limit does not depend on $\theta$, and hence the lower-tail behavior is unaffected by $\theta$.

\underline{If $-1\leq\eta<0$}, let $p=-\eta\in(0,1]$, and then $\psi(t)=\exp\{-H_p(t)\}$. Therefore, as $t\to\infty$, we have
\begin{align}
\kappa_L &=\lim_{t\to\infty}\frac{\log\psi(2t)}{\log\psi(t)} = \lim_{t\to\infty} \frac{-H_p(2t)}{-H_p(t)} = 2^p=2^{-\eta}, \notag \\[1mm]
\lambda_L &= \lim_{t\to\infty} \frac{\psi(2t)}{\psi(t)} = \lim_{t\to\infty} \exp\left\{-\left[H_p(2t)-H_p(t)\right]\right\} \notag \\
          & =\lim_{t\to\infty} \exp\left\{-\left[2^p-1+o(1)\right] H_p(t)\right\} = 0. \notag
\end{align}

\medskip

\underline{If $\eta=0$}, then $\psi(t)=\exp\{-H_0(t)\}$, where $H_0(t)=F_0(W_\theta(t))=\frac12[A(W_\theta(t))]^2$. Eqs.~\eqref{eq:W-infty} and \eqref{eq:A-refined} imply that, as $t \to \infty$,
\begin{align}
A\left(W_\theta(t)\right) &=\log\left(2W_\theta(t)\right)+O\left(W_\theta(t)^{-1}\right) =\log(2\beta_{\vert \theta\vert} t)+O(t^{-1}), \notag
\end{align}
and thus,
\begin{align}
\kappa_L = \lim_{t\to\infty} \frac{\log\psi(2t)}{\log\psi(t)} = \lim_{t\to\infty} \frac{H_0(2t)}{H_0(t)} = 1. \notag
\end{align}
Also, as $t\to\infty$, we have
\begin{align}
H_0(2t)-H_0(t) &=\frac12 \left[A\left(W_\theta(2t)\right)-A\left(W_\theta(t)\right)\right] \left[A\left(W_\theta(2t)\right)+A\left(W_\theta(t)\right)\right] \notag \\
&= \frac12 [ \log(2) + O(t^{-1}) ] [2 \log(2\beta_{\vert \theta\vert} t) + \log(2) + O(t^{-1}) ] \notag\\\
& = \log(2)\log(2\beta_{\vert \theta\vert} t) +  \frac12 [\log(2)]^2 + O(\log (t) / t) \to \infty. \notag
\end{align}
Therefore,
\begin{align}
\lambda_L = \lim_{t\to\infty} \frac{\psi(2t)}{\psi(t)} = \lim_{t\to\infty} \exp\left\{-\left[H_0(2t)-H_0(t)\right]\right\} = 0. \notag
\end{align}

\underline{If $0<\eta<1$}, let $p=\eta$ and $b=b_\eta=(1-p)p^{-2}$, and then, Eq.~\eqref{eq:Hpt-infty-proof} leads to
\begin{align}
\psi(t) &=\left[1+\frac{H_p(t)}{b}\right]^{-b} \sim b^bH_p(t)^{-b} \in\RV_{-pb}  =\RV_{-(1-p)/p}, \qquad t\to\infty. \notag
\end{align}
Therefore,
\begin{align}
\lambda_L &=\lim_{t\to\infty}\frac{\psi(2t)}{\psi(t)} =2^{-pb} =2^{-(1-\eta)/\eta}, \notag \\[1mm]
\kappa_L
&=\lim_{t\to\infty}\frac{\log\psi(2t)}{\log\psi(t)}
=1.
\notag
\end{align}
Combining the three cases proves the results for $\lambda_L(\eta)$ and $\kappa_L(\eta)$.

\medskip

\underline{Part IV: Upper tail:}

\medskip

Let $u=q(s)=1-\psi(s)$. Then $s=\psi^{-1}(1-u)$, and the survival copula $\widehat C(u,u)$ satisfies
\begin{align}
\widehat C(u,u) & =2u-1+C(1-u,1-u) =2[1-\psi(s)]-[1-\psi(2s)] =2q(s)-q(2s). \notag
\end{align}
Then,
\begin{align}
\lambda_U =\lim_{s\downarrow0}\frac{2q(s)-q(2s)}{q(s)}, \qquad \kappa_U =\lim_{s\downarrow0}
\frac{\log(2q(s)-q(2s))}{\log q(s)},\label{eq:upper-diagonal-limits-proof}
\end{align}
whenever the limits exist.

\medskip

\underline{If $-1\leq\eta<1$}, let $w=W_\theta(s)$ and $p=|\eta|$. Equation~\eqref{eq:F-zero-derivation} implies that
\begin{align}
F_p(w) &=w-\frac{1-p^2}{6}w^2+O(w^3), \quad s \downarrow 0. \notag
\end{align}

Recall that for $-1\leq\eta\leq0$, $G_\eta(w)=\exp\{-F_p(w)\}$, and for $0<\eta<1$, $G_\eta(w)=[1+F_p(w)/b_\eta]^{-b_\eta}$, where $b_\eta=(1-\eta)\eta^{-2}$. Applying the exponential and binomial expansions, respectively, we have, as $w \downarrow 0$, thus, as $F_p(w)\downarrow 0$,
\begin{align}
1-G_\eta(w) &= 
\begin{cases}
\displaystyle F_p(w)-\frac12F_p(w)^2+O(F_p(w)^3), & -1\leq\eta\leq0,\\[3mm]
\displaystyle F_p(w)-\frac{b_\eta+1}{2b_\eta}F_p(w)^2 +O(F_p(w)^3), & 0<\eta<1,
\end{cases} \notag \\[3mm]
&=
\begin{cases}
\displaystyle w-\left(\frac{1-p^2}{6} +\frac12\right)w^2+O(w^3), & -1\leq\eta\leq0,\\[3mm]
\displaystyle w-\left(\frac{1-p^2}{6} +\frac12+\frac{\eta^2}{2(1-\eta)}\right)w^2+O(w^3), & 0<\eta<1,
\end{cases} \notag \\[3mm]
&=w-c_2(\eta)w^2+O(w^3), \notag
\end{align}
where, because $(b_\eta+1)/(2b_\eta)=1/2+\eta^2/(2(1-\eta))$ and $p^2=\eta^2$,
\begin{align}
c_2(\eta) := \frac12 + \frac{(\eta_+)^2}{2(1-\eta)} +\frac{1-\eta^2}{6}>0, \qquad \eta_+ = \max(\eta,0). \notag
\end{align}

Because $\psi(s)=G_\eta(W_\theta(s))$,
\begin{align}
q(s) = 1 - \psi(s) = W_\theta(s)-c_2(\eta)W_\theta(s)^2 +O(W_\theta(s)^3). \label{eq:q-via-W-proof}
\end{align}

\underline{If $\theta=-1$}, Eq.~\eqref{eq:W-zero-all} shows $W_{-1}(s)=s$, and therefore, as $s \downarrow 0$, 
\begin{align}
q(s) &=s-c_2(\eta)s^2+O(s^3), \notag \\[1mm]
2q(s)-q(2s) &=2[s-c_2(\eta)s^2] - [2s-4c_2(\eta)s^2]+O(s^3) =2c_2(\eta)s^2+O(s^3), \notag
\end{align}
which proves $\lambda_U=0$ and $\kappa_U=2=1-\theta$.

\underline{If $-1<\theta<0$}, let $p=-\theta\in(0,1)$ and $c_p=p2^{1-p}$. Eq.~\eqref{eq:W-zero-all} shows $W_\theta(s)=s-c_ps^{1+p}+o(s^{1+p})$, and therefore, as $s \downarrow 0$, 
\begin{align}
q(s) &= \left[s-c_ps^{1+p}+o(s^{1+p})\right] -c_2(\eta)\left[s+o(s)\right]^2+O(s^3) =s-c_ps^{1+p}+o(s^{1+p}), \notag \\
2q(s)-q(2s) &=2(s-c_ps^{1+p}) - (2s-c_p2^{1+p}s^{1+p}) +o(s^{1+p}) =c_p(2^{1+p}-2)s^{1+p}
+o(s^{1+p}). \notag
\end{align}
Therefore,
\begin{align}
\lambda_U &=\lim_{s\downarrow0} \frac{2q(s)-q(2s)}{q(s)} =0, \notag \\[1mm]
\kappa_U  &=\lim_{s\downarrow0} \frac{\log(2q(s)-q(2s))}{\log q(s)} =1+p =1-\theta. \notag
\end{align}

\underline{If $\theta=0$}, let $\ell=\ell_s=\log(2/s)$. Eq.~\eqref{eq:W-zero-all} shows 
\begin{align}
W_0(s) &= s\left[1-\frac{1}{\ell}+\frac{1}{2\ell^2}+O(\ell^{-3})\right], \notag
\end{align}
together with Eq.~\eqref{eq:q-via-W-proof} and $s^2=o(s\ell^{-3})$, as $s \downarrow 0$,
\begin{align}
q(s) &=s\left[1-\frac{1}{\ell}+\frac{1}{2 \ell^2}+O(\ell^{-3})\right]. \label{eq:q-theta-zero-proof}
\end{align}

Since $\ell_{2s}=\log\left(2/(2s)\right)= \ell -\log(2)=: \ell - c$, as $s \downarrow 0$ and thus $\ell \to \infty$,
\begin{align}
2q(s)-q(2s) &=2s\left[\left(1-\frac{1}{\ell}+\frac{1}{2\ell^2}\right) -\left(1-\frac{1}{\ell-c}+\frac{1}{2(\ell-c)^2}\right) +O(\ell^{-3}) \right] \notag \\[1mm]
&=2s\left[\left(1-\frac{1}{\ell}+\frac{1}{2\ell^2}\right) -\left(1-\frac{1}{\ell}+\frac{1/2-c}{\ell^2}\right)
+O(\ell^{-3}) \right] \label{eq:upper-theta-zero-diagonal-intermediate-proof} \\[1mm]
&=\frac{2\log(2)s}{\ell^2} \left(1+O(\ell^{-1})\right), \label{eq:upper-theta-zero-diagonal-proof}
\end{align}
where Eq.~\eqref{eq:upper-theta-zero-diagonal-intermediate-proof} is due to $(\ell-c)^{-1}=\ell^{-1}+c\ell^{-2}+O(\ell^{-3})$ and $(\ell-c)^{-2}=\ell^{-2}+O(\ell^{-3})$. Because $q(s)\sim s$ as shown above, Eqs.~\eqref{eq:q-theta-zero-proof} and \eqref{eq:upper-theta-zero-diagonal-proof} imply that
\begin{align}
\lambda_U &=\lim_{s\downarrow0} \frac{2q(s)-q(2s)}{q(s)} =\lim_{s\downarrow0}\frac{2\log(2)}{\ell^2}=0, \notag \\[1mm]
\kappa_U &=\lim_{s\downarrow0} \frac{\log(s)-2\log \ell +O(1)}{\log(s)+o(1)} =1. \notag
\end{align}

\underline{If $0<\theta<1$}, let $a=1-\theta\in(0,1)$. Eq.~\eqref{eq:W-positive-zero} shows that $W_\theta\in\RV_a$ at zero and $W_\theta(s)\to0$. Eq.~\eqref{eq:q-via-W-proof} therefore leads to, as $s \downarrow 0$,
\begin{align}
q(s) &=W_\theta(s)\left[1+O(W_\theta(s))\right] \sim W_\theta(s), \notag
\end{align}
and therefore $q\in\RV_a$ at zero. Then, Eq.~\eqref{eq:upper-diagonal-limits-proof} implies that
\begin{align}
\lambda_U &=\lim_{s\downarrow0} \left(2-\frac{q(2s)}{q(s)}\right) =2-2^a =2-2^{1-\theta}, \notag
\end{align}
and thus, $2q(s)-q(2s) \sim(2-2^a)q(s)$, and then,
\begin{align}
\kappa_U &=\lim_{s\downarrow0} \frac{\log q(s)+\log(2-2^a)+o(1)}{\log q(s)} =1. \notag
\end{align}

Combining all four cases proves the results for the upper tail properties; together with those for the lower tail properties, this completes the proof of the proposition. \QED

\section{Additional Proofs}
\label{appendix:additional-proofs}
\emph{Proof of Proposition~\ref{prop:Kendall-tau}:} Let $(U,V)$ follow an Archimedean copula $C$ with generator $\psi$, $\phi:=\psi^{-1}$, then Kendall's distribution function of $W:=C(U,V)$ is given by
\begin{align}
K_C(w) :=  \P(W \le w) = w-\frac{\phi(w)}{\phi'(w)}. \notag
\end{align}
Therefore, the Kendall's tau can be expressed as
\begin{align}
\tau(C) & = 4\E(W)-1 = 4 \int_0^1[1-K_C(w)] d w -1 \notag \\
        & = 3-4\int_0^1K_C(w) d w = 1+4\int_0^1\frac{\phi(w)}{\phi'(w)} d w, \notag
\end{align}
and then a change of variables with  $w=\psi(t)$ and $\phi(w)=t$ proves the result. \QED

\emph{Proof of Proposition~\ref{prop:Spearman-rho}:} The definition of Spearman's rho is given by
\begin{align}
\rho_S =12 \int_0^1 \int_0^1 C(u,v) d u d v - 3. \notag
\end{align}
Then letting $u=\psi(x)$ and $v=\psi(y)$ proves the result. \QED

\emph{Proof of Proposition~\ref{prop:Blomqvist-beta}:} For a copula $C$, the Blomqvist's beta is defined as
\begin{align}
\beta_B =4C\left(1/2, 1/2\right)-1, \notag
\end{align}
which proves the result. \QED

\bibliographystyle{abbrv}
\bibliography{bib.bib}

\end{document}